\documentclass[11pt]{article}
\usepackage[utf8]{inputenc}	
\usepackage{amsmath,amsthm,amsfonts,amssymb,amscd}
\usepackage{multirow,booktabs}
\usepackage[table]{xcolor}
\usepackage{fullpage}
\usepackage{lastpage}
\usepackage{enumitem}
\usepackage{fancyhdr}
\usepackage{mathrsfs}
\usepackage{wrapfig}
\usepackage{setspace}
\usepackage{subcaption}
\usepackage{graphicx}
\usepackage{calc}
\usepackage{multicol}
\usepackage{cancel}
\usepackage[most]{tcolorbox}
\usepackage{xcolor}
\usepackage{hyperref}
\usepackage{tikz}
\usepackage{tikz-cd}
\usepackage{caption}
\usepackage{accents}
\usepackage{tikz,tikz-3dplot}
\usepackage{float}
\usepackage[normalem]{ulem}
\numberwithin{equation}{section}

\pdfoutput=1

\newtheorem{theorem}{Theorem}[section]
\newtheorem{assumption}[theorem]{Assumption}

\newtheorem{corollary}[theorem]{Corollary}
\newtheorem{definition}[theorem]{Definition}
\newtheorem{lemma}[theorem]{Lemma}
\newtheorem{proposition}[theorem]{Proposition}

\theoremstyle{remark}
\newtheorem{remark}[theorem]{Remark}
\newtheorem{example}[theorem]{Example}

\DeclareMathOperator{\supp}{supp}

\DeclareMathOperator*{\esssup}{ess\,sup}
\DeclareMathOperator*{\essinf}{ess\,inf}

\newcommand{\Acal}{\mathcal{A}}
\newcommand{\Bcal}{\mathcal{B}}

\newcommand{\Gcal}{\mathcal{G}}

\newcommand{\Lcal}{\mathcal{L}}
\newcommand{\Mcal}{\mathcal{M}}

\newcommand{\Pcal}{\mathcal{P}}

\newcommand{\Wcal}{\mathcal{W}}

\newcommand{\EE}{\mathbb{E}}

\newcommand{\NN}{\mathbb{N}}

\newcommand{\PP}{\mathbb{P}}

\newcommand{\RR}{\mathbb{R}}

\usepackage{bbm}

\renewcommand{\epsilon}{\varepsilon}

\newcommand{\sbm}{sBm\;}
\newcommand{\ssbm}{s$^2$Bm\;}

\newcommand{\CDF}{{\rm CDF}}
\newcommand{\CDFc}{{\rm CDF_c}}
\newcommand{\QF}{{\rm QF}}
\newcommand{\QFc}{{\rm QF_c}}

\DeclareMathOperator{\co}{co}

\DeclareMathOperator{\interior}{int}

\title{Calibration of the Bass Local Volatility model}

\author{Beatrice Acciaio\thanks{Department of Mathematics, ETH Z\"{u}rich, Switzerland.\newline \emph{beatrice.acciaio@math.ethz.ch}, \emph{antonio.marini@math.ethz.ch}, \emph{gudmund.pammer@math.ethz.ch}}, Antonio Marini\footnotemark[1], and Gudmund Pammer\footnotemark[1]}

\date{\today \vspace{-20pt}}

\begin{document}

\maketitle
\begin{abstract}
The Bass local volatility model introduced by Backhoff-Veraguas--Beiglb\"ock--Huesmann--K\"allblad is a Markov model perfectly calibrated to vanilla options at finitely many maturities, that approximates the Dupire local volatility model. Conze and Henry-Labord\`ere show that its calibration can be achieved by solving a fixed-point equation. In this paper we complement the analysis and show existence and uniqueness of the solution to this equation, and that the fixed-point iteration scheme  converges at a linear rate. 
\\\mbox{ }
\end{abstract}

\section{Introduction}\label{sect:intro}
In local volatility models, the instantaneous volatility of the asset price process is assumed to be a deterministic function of asset price and time. This eases the mathematical tractability while still encompassing a rich class of models.  
Among those, the Dupire local volatility (Dupire-LV) model \cite{Du94} is undoubtedly the most known and used one, largely due to the celebrated Dupire's formula, that enables us to deduce the volatility function from a surface of market prices  of
standard European options for a \emph{continuum} of maturities and strikes.
The Dupire-LV model represents the unique diffusion process consistent with the risk-neutral densities derived from the market prices of European options for all maturities. However, as it is not possible to observe a continuum of vanilla prices, interpolation is needed to approximate the prices and compute the local volatility numerically, a procedure that is sensitive to numerical instabilities and errors.

An appealing alternative to the Dupire-LV model has been introduced by Backhoff-Veraguas--Beiglb\"ock--Huesmann--K\"allblad~\cite{BaBeHuKa20}, and given the name of Bass local volatility (Bass-LV) model by Conze and Henry-Labord\`ere~\cite{CoHe21}.
This is a Markov model perfectly calibrated to the asset's risk-neutral marginal distributions $\mu_1, \dots, \mu_n$ implied from market prices of vanilla options for a \emph{finite set} of maturities $0\leq T_1< \dots <T_n$. Interestingly, one can prove that this model converges to the Dupire-LV model when $\sup_{i=1, \dots, n} |T_i-T_{i-1}|$ converges to $0$; see Corollary~\ref{cor.Dup} below.
\cite{CoHe21} shows that the Bass-LV model can be calibrated and simulated in a very efficient way (numerical experiments show that it can be 80 times faster than the Dupire-LV model) and requires only the simulation of a Brownian motion.
The procedure is based on a generalization of the Bass construction to the Skorokhod embedding problem. Specifically, given two distributions $\mu, \nu \in \Pcal (\RR)$ in convex order, we want to build a martingale $(M_t)_{t \in [0,1]}$ such that 
\begin{equation}\label{eq.ssbm_intro}
M_0 \sim \mu, \;\;\; M_1 \sim \nu, \;\;\; M_t = f_t(B_t),\; t\in[0,1],
\end{equation}
where $(B_t)_{t \in [0,1]}$ is a Brownian motion with initial distribution $\alpha=\alpha(\mu,\nu)$ to be determined, and $f_t(x):[0,1] \times \RR \rightarrow \RR$, $t\in[0,1]$, are suitable functions. 
This construction can be found in the literature under the name of \emph{Bass martingale} or \emph{standard stretched Brownian motion} (s$^2$Bm); see \cite{BaBeHuKa20} and Backhoff-Veraguas--Beiglb\"ock--Schachermayer--Tschiderer \cite{BaBeScTs23}.
It turns out that the process defined in \eqref{eq.ssbm_intro} constitutes the martingale that connects the two given marginals while staying as close as possible to the dynamics of a Brownian motion, see Theorem~\ref{thm.sbm} below. In particular, we have the original Bass construction when $\mu = \delta_{\text{mean}(\nu)}$.
Then the Bass-LV model can be defined as a \ssbm on each time interval $I_i=[T_i, T_{i+1}]$ between the marginals $\mu_i$ and $\mu_{i+1}$, $i=1, \dots, n-1$, up to scaling due to the length $T_{i+1}-T_i$ that may be different than $1$, see Definition~\ref{def:BLV} below.

In fact, the calibration of the \ssbm (cf.\ \eqref{eq.ssbm_intro}) requires solving the fixed-point equation proposed in \cite{CoHe21}, 
over the set $\CDF$ of cumulative distributions functions
\begin{equation}
\label{fixedPointEq_intro}
    F = \Acal F,
\end{equation}
where the nonlinear integral operator $\Acal:\CDF\to\CDF$ is defined as
\begin{equation}
\label{Acal_intro}
    \Acal F:= F_{\mu} \circ(\phi \ast ( Q_\nu \circ (\phi \ast F))).
\end{equation}
Here $\phi$ denotes the density of a standard normal distribution and, for any measure $\xi$, $F_\xi$ and $Q_\xi$ denote the corresponding CDF and quantile function.  
In particular, if a fixed-point $F$ for $\Acal$ exists, this  will be the CDF corresponding to the initial distribution $\alpha(\mu,\nu)$ of the Brownian motion for the \ssbm in \eqref{eq.ssbm_intro}.

In the present paper we complement the study performed in \cite{CoHe21}, under some technical conditions that we postpone to Assumption~\ref{ass:A1&2} below. 
We start by recalling the concept of irreducibility recently introduced in \cite[Definition 1.2]{BaBeScTs23}, that is necessary for the existence of a fixed point.

\begin{definition}[Irreducibility]
Let $\mu, \nu \in \Pcal(\mathbb R)$ with $\mu \preceq_c \nu$.\footnote{Here, $\preceq_c$ denotes the convex order relation for measures.} We say that the pair $(\mu, \nu)$ is irreducible if for all $A, B \in \mathcal B(\mathbb R)$ with $\mu(A), \nu(B) > 0$ there exists $\pi\in\Mcal(\mu,\nu):=\{\pi\in \Pcal(\mathbb R\times \mathbb R) : \text{$\pi$ has marginals $\mu,\nu$ and } \mbox{mean}(\pi^x)=x\; \mu\text{-a.e.}\}$ satisfying $\pi(A\times B)>0$.
\end{definition}
Roughly speaking, the pair $(\mu,\nu)$ is irreducible if the mass can be moved from anywhere in $\mu$ to anywhere in $\nu$ by means of some martingale coupling $\pi\in\Mcal(\mu,\nu)$, see \cite{BeJu16, BaBeScTs23}.
In \cite{BeJu16} it is shown that when the pair $(\mu, \nu)$ is not irreducible, the space $\mathbb R$ can be partitioned into a set on which $\mu$ and $\nu$ coincide as well as at most countably many sets, called irreducible components, so that restricted to each component, the pair $(\mu,\nu)$ is irreducible.
Clearly, irreducibility corresponds to having a unique component.

For every $\alpha \in \mathcal P(\mathbb R)$, we denote its CDF by $F_\alpha$ and its left-continuous quantile function by $Q_\alpha = F_\alpha^{-1}$.
Further, we recall that the $\infty$-Wasserstein distance between $\alpha_1,\alpha_2 \in \mathcal P(\mathbb R)$ precisely coincides with the $L^\infty$-distance of their quantile functions. Denoting by $\CDFc$ the set of cumulative distribution functions associated to distributions with compact support, we can state our first main contribution.
\begin{theorem}
\label{thm:regularity}
	Under Assumption~\ref{ass:A1&2},\footnote{More specifically, we assume here that the pair $(\mu,\nu)$ satisfies that $\nu$ is absolutely continuous w.r.t.\ the Lebesgue measure with density bounded away from zero on the convex hull of its support and the support of $\mu$ is contained in the interior of the support of $\nu$.} $\Acal$ has the following properties:
    \begin{enumerate}[label = (\roman*)]
        \item \label{it:L_inf_Thm1} {\rm Non-expansiveness}: $\Acal$ is non-expansive with respect to the $\infty$-Wasserstein distance\footnote{By abuse of notation, we write $\Wcal_\infty(F_\xi,F_\zeta) := \Wcal_\infty(\xi,\zeta)$ for distributions $\xi,\zeta \in \Pcal(\RR)$.}, i.e.,
        \[
        \Wcal_\infty(\Acal F, \Acal G) \le \Wcal_\infty(F,G)\quad \text{for all $F, G \in {\rm CDF}_c$}.
        \]
        In particular, $\Wcal_\infty(\Acal F, \Acal G)= \Wcal_\infty(F, G)$ if and only if $F^{-1}-G^{-1} \in {\rm Span}(\mathbf 1) \subseteq L^\infty(0,1)$\footnote{In other words, the operator $\Acal$ is a contraction in $L^\infty(0,1)$ in any direction except along $\mathbf 1$, where $\mathbf 1$ is the function in $L^\infty(0,1)$ which is constantly equal 1. We denote by Span$(\mathbf 1)$ the subspace $\{c\mathbf 1 : c \in \RR \}$.}.
        \item \label{it:L_inf_Thm2} {\rm Uniqueness}: There exists at most one (up to translation) fixed point $F^* \in \CDFc$ of $\Acal$.
        \item \label{it:L_inf_Thm3} {\rm Existence}: $\Acal$ admits a fixed-point if and only if $(\mu,\nu)$ is irreducible.
    \end{enumerate}
\end{theorem}
The proof of this theorem relies on the properties of an auxiliary operator $\Gcal$, which can be viewed as the representation of the operator $\Acal$ on the space of quantile functions, see Proposition~\ref{mainProposition}. Notably, our proof leverages the fact that the operator $\Gcal$ is Fréchet differentiable and non-expansive with respect to the supremum norm.
Furthermore, we can establish the non-expansivness of the operator $\Acal$ under weaker assumptions as a corollary of Theorem \ref{thm:regularity}.

\begin{corollary}
\label{cor:non-exp}
    If $\nu$ is compactly supported and $\supp(\mu) \subseteq {\rm int}({\rm co}({\rm supp}(\nu)))$, then $\Acal$ is non-expansive with respect to the $\infty$-Wasserstein distance.
\end{corollary}

Under the assumption of irreducibility, we further obtain convergence of the fixed-point iteration scheme for any initial starting CDF, confirming what the numerical results presented in \cite{CoHe21} suggest. 
This constitutes our second main contribution.
\begin{theorem}\label{thm:conv}
    Under Assumption \ref{ass:A1&2}, let $(\mu,\nu)$ be irreducible and $F \in \CDF$.
    Then $(\Acal^k F)_{k \in \NN}$ converges to a fixed-point $F^\ast$ of $\Acal$ and there is $q \in (0,1)$ such that:
    \begin{enumerate}[label = (\roman*)]
        \item \label{it:conv_1} {\rm Linear convergence\footnote{ Let $(X,d)$ be a metric space. A sequence $(x_n)_{n \in \NN} \subseteq X$  converges linearly to $x\in X$ if there exists $q\in (0,1)$ such that $d(x, x_{n+1}) < q d(x, x_{n})$, for any $n \in \NN$.} to solution space:} For all $k \in \NN$ we have $\Wcal_\infty(\Acal^{k + 1}F, \Lcal^\ast) \le q \Wcal_\infty(\Acal^{k}F, \Lcal^\ast)$, where $\Lcal^\ast$ is the set of fixed-points of $\Acal$;
        \item \label{it:conv_2} {\rm Convergence:} For all $k \in \NN$ we have $\Wcal_\infty(\Acal^k F,F^\ast) \le 2 q^k \Wcal_\infty(F,F^\ast)$.
    \end{enumerate}
\end{theorem}
The proofs of Theorem~\ref{thm:regularity}, Corollary~\ref{cor:non-exp}, and Theorem~\ref{thm:conv} are presented in Section~\ref{sect:main_proofs}. Under additional assumptions, we are able to show a stronger result, that is, that the iteration scheme converges linearly, see Proposition~\ref{prop:lin}.

\begin{remark}[Irreducibility]
    At a first glance the assumption on irreducibility of the marginals in Theorem \ref{thm:conv} may seem restrictive.
    But as discussed earlier, it is always possible to partition $\RR$ into irreducible components.
    In one dimension this decomposition can be effortlessly computed thanks to the concept of  potential functions, see \cite[Appendix A.1]{BeJu16} for a detailed discussion.
    Moreover, since restricted to each component the \sbm from $\mu$ to $\nu$ coincides with the \ssbm on said component (cf. Theorem~\ref{thm.sbm} below), the fixed-point algorithm can be applied there in order to find the sBm.
\end{remark}

\begin{remark}[Compact support]
    To establish Theorem \ref{thm:regularity} and Theorem \ref{thm:conv}, we impose in Assumption \ref{ass:A1&2}  certain regularity properties on $\mu$ and $\nu$.
    In light of stability of the Martingale Benamou-Brenier problem \cite[Corollary 1.4]{BeJoMaPa21b} w.r.t.\ the marginal inputs, this assumption is not excessively restrictive.
    Indeed, thanks to Lemma \ref{lem:approx.sequence} we can approximate any pair $(\mu,\nu)$ in the convex order with a sequence of pairs $(\mu_n,\nu_n), {n\in\NN}$, that are irreducible and satisfy Assumption~\ref{ass:A1&2}.
    It follows from the stability result in \cite[Corollary 1.4]{BeJoMaPa21b} that the \sbm from $\mu_n$ to $\nu_n$ converges in law (as a distribution on $C([0,1])$) to the \sbm from $\mu$ to $\nu$.
\end{remark}

While the analysis in the current paper is presented in the one-dimensional case as considered in \cite{CoHe21}, we stress that the fixed-point iteration can be extended to the general multi-dimensional case. 
The extension of our results to this case is significantly more intricate both at a conceptual and technical level, and deserves a study of its own.

\paragraph{Other related literature.} 
{As discussed above, 
the model studied in the present paper closely relates to the Dupire-LV model. However, it is worth mentioning other methods for local volatility calibration, obtained for example via Tikhonov regularization (see \cite{crepey2003calibration,de2012convex}),
via entropy minimization (see \cite{avellaneda1997calibrating}), via Monte Carlo methods (see \cite{henry2009calibration}), and via martingale optimal transport (see \cite{GuLoWa19}). 
The problem we consider is in fact a martingale optimal transport problem, formulated in a dynamical setting.
}
Optimal transport theory dates back to Monge \cite{Monge} and Kantorovich \cite{Kant42}, with a prominent role in its modern formulation played by the seminal works \cite{Br87, BB00, Mc95, JoKiOt98}; we refer to the manuscript \cite{Vi09} for an overview. 
The current paper focuses on a martingale optimal transport problem, where transport plans have to satisfy the martingale constraint. 
This class of problems plays a key role in financial modeling (see, e.g., \cite{HoNe12, BeHePe12, HeTaTo16,  DoSo12, AcBePeSc16, He17, GuLoWa19, NuWiZh22}), and is inherently interesting from a probabilistic perspective due to connections to martingale inequalities (see, e.g., \cite{BeSi13, AcBePeScTe13, HeObSpTo16}) and the Skorokhod embedding problem (see, e.g., \cite{BeCoHu17, KaTaTo15, BeNuSt19}).

Another parallel of martingale optimal transport to classical optimal transport, is established by the Martingale Benamou-Brenier problem introduced in \cite{HuTr17,BaBeHuKa20}.
More recently, \cite{BaBeScTs23, BaScTs23} study the structure of this problem from a dual perspective.
An early work of Loeper~\cite{Lo18} considers a calibration problem for option pricing with linear market impact from a PDE perspective, which was later extended in \cite{GuLoWa19} to a setting closely related to the current one.
On the algorithmic side, the initial study was done in \cite{CoHe21}, where the connection of the fixed-point equation \eqref{fixedPointEq_intro} and the Bass construction was made and an efficient way of computing the option Greeks provided.
A novel way of computing the fixed-point iteration of \cite{CoHe21} is proposed in \cite{JoLoOb24} based on PDE methods.
The role of irreducibility in martingale optimal transport problems was initially studied in \cite{BeJu16} and then extended in \cite{DeTo17, ObSi17, BaBeScTs23}.

\paragraph{Outline of the paper.} The rest of the paper is organized as follows. We conclude the Introduction with a list of notations that will be used throughout the paper.
In Section \ref{sect:theory} we provide the theoretical background. In particular, we introduce the standard stretched Browniam motion, the Martingale Benamou-Brenier formula and the Bass-LV model. Here we recall the interplay between those concepts and show the connections between the Bass-LV and the Dupire-LV model. 
Section \ref{sect:compact} is devoted to showing our main contributions. These include the proofs of Theorems~\ref{thm:regularity} and~\ref{thm:conv}, as well as a study of the fixed-point operator, for which we show geometric properties and regularity.
Finally, 
in Section~\ref{sect:numerics} we present some numerical illustrations of the theoretical results and remarks for the semidiscrete problem.\\

\paragraph{Notations.}
\begin{itemize}
\item For $d\in\NN$, we write $\Pcal(\RR^d)$ for the probability measures on $\RR^d$ and $\Pcal_p(\RR^d)$ for the subset of probability measures  with finite $p$-moment, $p\in[1,\infty)$.
\item We use CDF as abbreviation for cumulative distribution function and, with an abuse of notation, we also denote by CDF the set of all cumulative distribution functions. Moreover, we write $\CDFc$ for the subset of cumulative distribution functions associated to distributions with compact support.
\item We denote by $\QF$ the set of quantile functions and by $\QFc$ the subset of bounded quantile functions.
\item For any  $\xi\in\Pcal(\RR)$, we use $F_\xi$ and $Q_\xi$ to denote the corresponding CDF and quantile function.
We write mean$(\xi)=\int  y \xi(dy)$ for its mean and denote by $\supp(\xi)$ its support and by $\co(\supp(\xi))$ the closed convex hull of the support.
\item For every $t>0$, we denote by $\gamma_t$ a centered Gaussian with variance $t$ and by $\phi_t$ its density. As we often use the standard normal distribution, to relax the notation we set $\gamma=\gamma_1$, $\phi=\phi_1$, and as usual we denote by $\Phi$ its CDF. Moreover, we write $\lambda^d$ for the Lebesgue measure on $\RR^d$, and again simplify the notation by setting $\lambda=\lambda^1$.
\item For $\mu,\nu\in\Pcal(\RR)$, we denote by  $\Pi(\mu,\nu)$ the subset of $\Pcal(\RR\times\RR)$ of measures with first marginal $\mu$ and second marginal $\nu$. The elements of $\Pi(\mu,\nu)$ are called couplings of $\mu$ and $\nu$. We use
$\Mcal(\mu,\nu)$ for the subset of $\Pi(\mu,\nu)$ containing the measures $\pi$ such that $\mbox{mean}(\pi^x)=x\; \mu\text{-a.e.}$, where
$\pi^x$ is the regular conditional disintegration of $\pi$ w.r.t. $\mu$: $\pi(dx,dy)=\mu(dx)\pi^x(dy)$.
The elements of $\Mcal(\mu,\nu)$ are called martingale couplings of $\mu$ and $\nu$.  
\item We denote by $\Bcal(\RR)$ the Borel sets of $\RR$.
\item The push-forward measure of $\xi \in \Pcal (\RR)$ through a measurable map $T: \RR \rightarrow \RR$, denoted by $T_\# \xi$, is the probability measure such that $T_\#\xi(A)=\xi(T^{-1}(A))$, for any $A\in\Bcal(\RR)$.
\item For $\mu,\nu\in\Pcal_1(\RR)$, we say that $\mu$ is dominated in convex order by $\nu$, and write $\mu\preceq_c\nu$, if for all convex functions $h:\RR\to\RR$ with linear growth we have $\int h d\mu\leq \int h d\nu$.
\item For $\xi,\zeta\in\Pcal(\RR)$, we write $\xi\ast\zeta$ for the probability measure representing their convolution, so that $\xi\ast\zeta(A)=\int 1_A(x+y)d\xi(x)d\zeta(y)$, $A\in\Bcal(\RR)$. For two measurable functions $f,g$,
their convolution is the function given by $f\ast g(x)=\int f(x-y)g(y)dy$, $x\in\RR$. Moreover, the convolution of $f$ and $\xi$ is the function defined as $f\ast \xi(x)=\int f(x-y)\xi(dy)$, $x\in\RR$.
\item For $A\subseteq\RR$, with $L^\infty(A)$ we mean the set of $\lambda$-essentially bounded functions $f:A\to\RR$, and with an abuse of notation we write $L^\infty(a,b):=L^\infty((a,b))$, $a,b\in\RR$.
We equip $L^\infty(A)$ with convergence w.r.t.\ the $\lambda$-essential supremum norm $\|\cdot\|_\infty$.
Similarly, 
$L^0(A)$ is the set of measurable functions $f:A\to\RR$, and $L^0(a,b):=L^0((a,b))$. We equip $L^0(A)$ with convergence in probability w.r.t.\ Lebesgue measure. Finally, for any $\xi \in \Pcal(\RR)$, $L^1(A; \xi)$ is the set of functions $f:A \rightarrow \RR$ such that $\|f\|_{L^1(A; \xi)} :=\int_A |f(x)|d\xi(x) < \infty$, and $L^1(A):= L^1(A; \lambda)$. We equip $L^1(A; \xi)$ with convergence w.r.t. the $L^1$ norm $\|\cdot\|_{L^1(A; \xi)}$.
We write $\bf 1$ for the function that is constant equal 1 in $L^\infty(A)$.
\item For $p\in[1,\infty)$, the $p$-Wasserstein distance between two probability measures $\xi,\zeta\in\Pcal_p(\RR)$ is given by
\[
\Wcal_p(\xi,\zeta) :=\inf_{\pi\in\Pi(\xi,\zeta)}\left(\int |x-y|^p\pi(dx,dy)\right)^{1/p}.
\]
The $\infty$-Wasserstein distance, obtained as limit for $p\to\infty$ of the $p$-Wasserstein distance, is then given by 
\[
\Wcal_\infty(\xi,\zeta) :=\inf_{\pi\in\Pi(\xi,\zeta)}\pi-\esssup |x-y|.
\]
By abuse of notation, we write $\Wcal_\infty(F_\xi,F_\zeta) := \Wcal_\infty(\xi,\zeta)$.
\end{itemize}

\section{State of the art and connections}\label{sect:theory}
This section is devoted to introducing the Bass local volatility model and other main concepts related to it.
The core problem is to find a martingale in continuous time that connects two given distributions $\mu,\nu \in \Pcal(\RR)$. 
This clearly requires more than simply finding a martingale coupling of $\mu$ and $\nu$, that is an element $\pi\in\Mcal(\mu,\nu)$. 
Indeed, the latter problem is only concerned with the static perspective of defining a joint distribution of $\mu$ and $\nu$ that constitutes a 2-step martingale, while we are also interested in how the martingale evolves in continuous time between those two distributions.
In recent mathematical finance literature, a lot of effort has been devoted to the static martingale optimal transport (see, e.g., \cite{BeHePe12, CaLaMa14, DoSo14, GaHeTo13, AcBePeSc16, HoNe12, TaTo13}), which is clearly a source of ideas when one looks at the dynamic problem. 
Crucially, by Strassen theorem \cite{St65}, a martingale coupling between two distributions exists if and only if these are in convex order. This entails that if said condition fails, we will not have any solution to the dynamic problem either. So for the rest of the section we will assume $\mu\preceq_c\nu$. 

 \subsection{Stretched Brownian motion and standard stretched Brownian motion}
In classical optimal transport theory, the celebrated Benamou-Brenier formula offers a dynamic formulation of the problem for quadratic cost, see \cite{BB00}. Here the transport is interpreted as particles moving along absolutely continuous paths, from an initial configuration given by the first marginal $\mu$ until reaching the final configuration according to the second marginal $\nu$, so that
\begin{equation}\label{eq.bb}
\Wcal_2(\mu,\nu)^2=	\inf_{\substack{X_t = X_0 + \int_0^t \sigma_s ds \\ X_0 \sim \mu, X_1 \sim \nu}} \mathbb E \left[\int_0^1 |\sigma_t|^2 dt\right],
\end{equation}
where $(\sigma_s)_{s \in [0,1]}$ is an $\mathbb R$-valued stochastic process. This means that the particles are moving from $\mu$ to $\nu$ while staying as close as possible to a constant-speed particle.
For $\mu,\nu$ in convex order, 
its martingale counterpart, coined as Martingale Benamou-Brenier formula, reads as
\begin{equation}\label{eq.mbb}
	{\rm MBB}(\mu, \nu) 
 = \inf_{\substack{M_t = M_0 + \int_0^t \sigma_s dBs \\ M_0 \sim \mu, M_1 \sim \nu}} \mathbb E \left [ \int_0^1 |\sigma_t-1|^2 dt \right],
\end{equation}
where the optimization is taken over the class of filtered probability spaces $(\Omega, \mathcal F, P)$, with an $\mathbb R$-valued $\mathcal F$-progressive measurable process $(\sigma_t)_{t \in [0,1]}$  and an $\mathcal F$-Brownian motion $B$, such that $M$ is a martingale. 
As shown in \cite{BaBeHuKa20}, the choice of the underlying probability space is not relevant, provided that $(\Omega, \mathcal F, P)$ is rich enough to support an $\mathcal F_0$-measurable random variable with continuous distribution.
A process solving \eqref{eq.mbb} is called \emph{stretched Brownian motion} (sBm), to indicate its property of mimicking as closely as possible the movement of a Brownian particle while fitting the given marginals.
For example, when we have log-normal distributions in convex order, the sBm is simply the geometric Brownian motion between them.

The concept of stretched Brownian motion is closely related to that of a Bass martingale or standard stretched Brownian motion; see Theorem~\ref{thm.sbm} below and \cite{BaBeHuKa20} and \cite{BaBeScTs23} for a  thorough analysis. Here we present in more details the concept of \ssbm introduced in a simple way in
\eqref{eq.ssbm_intro}.

\begin{definition}[Standard stretched Brownian motion]
\label{s2Bm}
Let $\mu, \nu \in \Pcal (\RR)$. A martingale $(M_t)_{t \in [0,1]}$ is a standard stretched Brownian motion (s$^2$Bm) from $\mu$ to $\nu$ 
 if $M_0 \sim \mu$, $M_1 \sim \nu$, and there exist a non-decreasing function $f:\RR \rightarrow \RR$ and a Brownian motion $B=(B_t)_{t \in [0,1]}$ with a possibly non-trivial initial distribution such that $M_t = (\phi_{1-t} \ast f)(B_1)$.
\end{definition}

Note that if $\alpha \in \Pcal(\RR)$ and $f$ are the initial distribution of the Brownian motion and the non-decreasing map introduced Definition \ref{s2Bm}, respectively, then 
\begin{equation}
    \label{fund_rel}
    ({f})_\#(\gamma_1\ast\alpha) = \nu \quad \text{and} \quad (\phi_{1-t} \ast f)_\#\alpha = \mu.
\end{equation}
Furthermore, observe that when $\mu=\delta_x$ this corresponds to the Bass construction of the solution to the Skorokhod embedding problem, and in this case we can take $\alpha=\delta_0$; see \cite{Ba83}.

\begin{theorem}[{\cite[Theorem 1.5]{BaBeHuKa20} and \cite[Theorem 1.3]{BaBeScTs23}}, Existence and uniqueness of the sBm]\label{thm.sbm}
    ~ \\
    Let $\mu, \nu \in \Pcal_2(\mathbb R)$ be such that $\mu\preceq_c\nu$. Then: 
    \begin{itemize}
    \item[(i)] There exists a \sbm from $\mu$ to $\nu$ (i.e. an optimizer for ${\rm MBB}(\mu, \nu)$) and is unique in law;
    \item[(ii)] There exists a \ssbm from $\mu$ to $\nu$ iff $(\mu, \nu)$ is irreducible, and in this case the \sbm is a s$^2$Bm.
    \end{itemize}
\end{theorem}

It was established in \cite{BaBeHuKa20} that the \sbm from $\mu$ to $\nu$ is a \ssbm on each irreducible component, thus the \sbm can be determined by finding the \ssbm on each component.

Now recall that Brenier's theorem states that, for $\xi,\zeta\in\Pcal_2(\RR)$ with $\xi$ absolutely continuous, the unique 
optimizer $\pi^*\in\Pi(\xi,\zeta)$ for the optimal transport with quadratic cost is concentrated on the graph of a monotone function $T_{\xi,\zeta}:\RR\to\RR$, so that $\pi^*=(Id,T_{\xi,\zeta})_\#\xi$. This means that $\Wcal_2(\xi,\zeta)^2=\int |x-T_{\xi,\zeta}(x)|^2\xi(dx)$, with $T_{\xi,\zeta}(x)=Q_\zeta\circ F_\xi$. Then, from Definition~\ref{s2Bm} and \eqref{fund_rel}, we see that a \ssbm from $\mu$ to $\nu$ exists if and only if there exists a distribution $\alpha \in \Pcal(\RR)$ such that the following diagram commutes 
\[
\begin{tikzcd}[row sep=0.7in, column sep = 1.4in]
  M_0 \sim \mu \arrow[r, "\text{\normalsize Brownian martingale}"]  & M_1 \sim \nu  \\
  B_0 \arrow[u, "\text{\normalsize $T_{\alpha, \mu}$}"] \sim \alpha \arrow[r, "\text{\normalsize Brownian motion}"] & B_1 \sim \alpha \ast \gamma  \arrow[u, "\text{\normalsize $T_{\alpha\ast \gamma, \nu}$}"]
\end{tikzcd}
\]
that is if and only if
\begin{equation}\label{eq.comm}
    T_{\alpha, \mu} = \phi \ast T_{\alpha \ast \gamma, \nu}.
\end{equation}
In particular, computing the \ssbm from $\mu$ to $\nu$ is equivalent to  finding $\alpha \in \Pcal (\RR)$ such that \eqref{eq.comm} holds.

\subsection{The Bass-LV model and the calibration problem} 
The Bass-LV model can be easily, quickly and perfectly calibrated to vanilla options. In fact, this model requires only the simulation of a Brownian motion and numerical experiments in \cite{CoHe21} showed that it can be 80 times faster than the Dupire local volatility model.

The Bass-LV model can be defined as follows.

\begin{definition}[Bass-LV model]\label{def:BLV}
Let $T>0$ and $\mu_1, \dots, \mu_n$ be the asset's risk-neutral marginal distributions implied from market price of vanilla options with expiries $0\leq T_1< \dots <T_n=T$. The Bass-LV model is a \ssbm on each time interval $I_i=[T_i, T_{i+1}]$, $i=1, \dots, n-1$, up to a time change as the length $T_{i+1}-T_i$ may differ from $1$. 
\end{definition}

Therefore, the calibration of the Bass-LV model is equivalent to computing a standard stretched Brownian motion from $\mu_i$ to $\mu_{i+1}$, for any $i=1, \dots, n-1$, accounting for scaling. 

From \eqref{eq.comm} and \cite[Proof of Theorem 2.1]{CoHe21} we have the following result.

\begin{theorem}[\cite{CoHe21}]
\label{char_s2Bm}
Let $\mu, \nu \in \mathcal{P}_2(\mathbb R)$ be such that $\mu \preceq_c \nu$. Then, there exists a standard stretched Brownian motion between $\mu$ and $\nu$ if and only if there exists a fixed point $F^*$ for the nonlinear integral operator $\Acal$ defined in \eqref{Acal_intro}.
In this case, the standard stretched Brownian motion from $\mu$ to $\nu$ is given by $M_t = f_t(B_t)$, where 
 \begin{equation*}
     f_t = \phi_{1-t} \ast (Q_\nu \circ (\phi \ast F^*)),
 \end{equation*}
 and $(B_t)_{t \in [0,1]}$ is a Brownian motion with $B_0 \sim \alpha$, where $\alpha$ is such that $F_\alpha=F^*$.
\end{theorem}
In order to do this for the pair $(\mu_i,\mu_{i+1})$ in the interval $[T_i,T_{i+1}]$, we again need to adapt by scaling, thus finding fixed points for the operators
\begin{equation}
\label{Acal_i}
    \Acal_i F:= F_{\mu_i} \circ(\phi_{T_{i+1}-T_i} \ast ( Q_{\mu_{i+1}} \circ (\phi_{T_{i+1}-T_i} \ast F))).
\end{equation}

\subsection{Connections with Dupire local volatility model}
In this section, we follow closely \cite[Section 1.4.6]{BaBeHuKa20} and provide the detailed argument for the convergence to the Dupire-LV model.
Before diving into the analysis of the fixed-point scheme we want to discuss in which sense the Bass-LV model is an approximation of the local volatility model.
For this reason, we consider a curve of marginal distributions $(\mu_t)_{t \in [0,T]}$, increasing in convex order, and a sequence of partitions $0 = T_0^{(n)} < T_1^{(n)} < \ldots < T_n^{(n)} = T$, $n \in \mathbb N$, where $\sup_{i = 1,\ldots, n} |T_i^{(n)} - T_{i - 1}^{(n)}|$ vanishes for $n \to \infty.$
To give a proper construction, let $B = (B_t)_{t \in [0,T]}$ be a Brownian motion and $U$ be independent and uniformly distributed on $(0,1)$.
By \cite[Theorem 2.2 and Remark 2.3]{BaBeHuKa20} the stretched Brownian motion between any two marginals can be written as a function of starting position and an independent Brownian motion.
Building on this observation, we obtain measurable maps $f^{(n)}_t \colon \mathbb R \times \mathbb R \to \mathbb R$ for $t \in [0,T]$ such that
\[
    f^{(n)}_t \left(F^{-1}_{\mu_{T^{(n)}_{i - 1}}}(U), B_t - B_{T^{(n)}_{i - 1}}\right)
\]
is a stretched Brownian motion starting at time $T^{(n)}_{i - 1}$ in $\mu_{T^{(n)}_{i - 1}}$ and terminating at time $T^{(n)}_{i}$ in $\mu_{T^{(n)}_{i}}$.
Successively for each $i = 1,\ldots,n$, we define the process $M^{(n)} = (M^{(n)}_t)_{t \in [0,T]}$ via
\[
    M_0^{(n)} := F_{\mu_0}^{-1}(U) \text{ and }M^{(n)}_t := f_t^{(n)} \left(M^{(n)}_{T^{(n)}_{i-1}},B_t - B_{T^{(n)}_{i - 1}} \right),
    \quad t \in [T^{(n)}_{i - 1}, T^{(n)}_{i}].
\]
By construction $M^{(n)}$ is a continuous, strong Markov martingale with prescribed one-dimensional marginal distributions on the $n$-th partition.

We recall \cite[Definition 1.1]{Lo08c} that an almost continuous diffusion $X$ is a real-valued strong Markov process with c\`adl\`ag paths that are right-continuous in probability such that for any iid pair $Y, Z \sim X$ and $0 < s < t$
\[
    \PP[Y_s < Z_s, Y_t > Z_t \text{ and }Y_u \neq Z_u \text{ for every }u \in (s,t) ] = 0.
\]

\begin{theorem}
    Assume that $(\mu_t)_{t \in [0,T]}$ is weakly continuous, increasing in convex order, and $\mu_T \in \mathcal P_1(\mathbb R)$.
    Then $M^{(n)}$ converges in finite dimensional distributions to the unique almost continuous diffusion $M$ with one-dimensional marginal distributions $(\mu_t)_{t \in [0,T]}$.
\end{theorem}

\begin{proof}
    Recall that, for every $n \in \mathbb N$, $M^{(n)}$ is a martingale with continuous paths that terminates in $\mu_T$.
    As a consequence of \cite[Theorem 4 and Theorem 11]{MeZh84}, we thus have that any subsequence of $(M^{(n)})_{n \in \mathbb N}$ admits subsequences that converge in finite dimensional distributions.
    Furthermore, it follows from the martingale convergence theorem that any martingale $X = (X_t)_{t \in [0,T]}$ with $\EE[|X_T|] < \infty$ and such that $t \mapsto Law(X_t)$ is weakly continuous, is also continuous in probability.
    Therefore, we can invoke \cite[Corollary 1.3]{Lo08c} to find that each such limit is an almost-continuous diffusion, cf.\ \cite[Definition 1.1]{Lo08c}.
    On the other hand, \cite[Theorem 1.3]{Lo08b} tells us that there exists a unique almost continuous diffusion with one-dimensional marginals $(\mu_t)_{t \in [0,T]}$.
    Hence, $(M^{(n)})_{n \in \mathbb N}$ converges in finite dimensional distributions to this process.
\end{proof}

As immediate consequence we get the next corollary.

\begin{corollary}\label{cor.Dup}
    The sequence $(M^{(n)})_{n \in \mathbb N}$ converges in finite dimensional distributions to the Dupire local volatility model if the latter exists and is a strong Markov martingale.
\end{corollary}

The idea of using \sbm to interpolate between finitely many marginals and then passing to a limit, was picked up earlier in \cite{BaBeHuKa20, BePaSc21, PaRoSc22}.
In \cite{PaRoSc22} an existence result in the vein of Corollary \ref{cor.Dup} was established in a regularized setting while the authors also provide counterexamples for uniqueness.
Nevertheless, we insist that the question of an appropriate multi-dimensional counterpart to Corollary \ref{cor.Dup} is still an open and interesting endeavour.

\section{Analysis of the fixed-point iteration}\label{sect:compact}
This section is devoted to the study of the fixed-point iteration scheme introduced in \cite{CoHe21}. In particular, to the proofs of Theorem~\ref{thm:regularity} and Theorem~\ref{thm:conv}, that collect our main contributions.

From now on, we fix measures $\mu,\nu \in \mathcal P(\mathbb R)$ and let $\mathcal A$ be the nonlinear integral operator in \eqref{Acal_intro}.
\begin{assumption}\label{ass:A1&2}
    The measures $\mu$ and $\nu$ satisfy:
    \begin{enumerate}[label = (A\arabic*)]
        \item \label{it:A1} The measure $\nu$ is absolutely continuous w.r.t.\ the Lebesgue measure;
        \item \label{it:A2} The support of $\mu$ is contained in ${\rm int}({\rm co}({\rm supp}(\nu)))$;
        \item \label{it:A3} The density of $\nu$ is bounded away from zero on $\co(\supp(\nu))$.
    \end{enumerate}
\end{assumption}
Note that \ref{it:A3} implies that $\nu$ has compact support and the derivative of $Q_\nu$ satisfies $\lambda$-almost surely that $0 \le Q'_\nu \le L$ for $L := \sup_{u \in (0,1)} \frac{1}{F_\nu' \circ Q_\nu(u)}$.

\subsection{Basic properties}\label{ssec:basic}
We start by establishing some basic properties of $\Acal$.
As preparation, consider $Q \in L^0(0,1)$ and let $F$ be the CDF of the random variable $Q(U)$ where $U$ is uniformly distributed on $(0,1)$.
Observe that
\begin{equation}
    \label{eq:Acal.innerpart}
    \phi \ast F(x) = \int_\RR \phi(x- y) F(y) \, dy =  \int_{\RR} \Phi(x - y) \, dF(y) = \int_0^1 \Phi(x - Q(y)) \, dy.
\end{equation}
The right-hand side in \eqref{eq:Acal.innerpart} inherits from $\Phi$ the property of being a smooth, strictly increasing bijection from $\RR$ to $(0,1)$.
Thus, by composing $Q_\nu$ with \eqref{eq:Acal.innerpart} we obtain an increasing function on $\RR$ taking values in $\supp(\nu)$ with
\begin{align} \label{eq:Acal.innerpart.lim-}
    \lim_{x \to -\infty} Q_\nu  \left( \int_0^1 \Phi(x - Q(y)) \, dy \right) = \inf_{u \in (0,1)} Q_\nu(u), \\
    \label{eq:Acal.innerpart.lim+}
    \lim_{x \to +\infty} Q_\nu  \left( \int_0^1 \Phi(x - Q(y)) \, dy \right) = \sup_{u \in (0,1)} Q_\nu(u).
\end{align}

We say that a measure $\eta \in \Pcal(\RR)$ is symmetric if its quantile function $Q_\eta$ is symmetric around $1/2$, i.e., $Q_\eta(1/2 + u) = Q_\eta(1/2 - u)$ for all $u \in (0,1/2)$. Similarly, we call a CDF $F$ symmetric if the associated distribution is symmetric.

\begin{lemma} \label{lem:Acal.properties}
    Let $F, G \in \CDF$ and $\nu \in \mathcal P(\RR)$ be compactly supported. The operator $\mathcal A$ satisfies:
    \begin{enumerate}[label = (\roman*)]
        \item \label{it:A.range} {Range}: Assuming \ref{it:A2}, we have ${\rm range}(\mathcal A) \subseteq {\rm CDF_c}$.
        \item \label{it:A.shift} {Shift-invariance}: $(\mathcal A F(\cdot + c))(x)=(\mathcal A F)(x + c)$ for every $c \in \mathbb R$.
		\item \label{it:A.monotone} {Monotonicity}: If $F \leq G$, then $\mathcal A F \leq \mathcal A G$. 
        \item \label{it:A.symmetry} {Symmetry}: If $\mu$, $\nu$, $F$ are symmetric so is $\Acal F$.
        Further, the medians of $F$ and $\Acal F$ coincide.
    \end{enumerate}
\end{lemma}

\begin{proof}
We show the first statement in \ref{it:A.range}.
Let $F \in {\rm CDF}$ and denote its quantile function by $Q$.
Write $H$ for $Q_\nu$ composed with the right-hand side in \eqref{eq:Acal.innerpart} and note that, as $\nu$ is compactly supported by assumption, $H$ is bounded.
Since $F_\mu$, $H$ as well as $\phi \ast H$ are all increasing functions, the same holds true for $\mathcal AF = F_\mu \circ (\phi \ast H)$.
Further, since $\phi \ast H$ is continuous and $F_\mu$ is right-continuous, we conclude that $\mathcal AF$ is also right-continuous.
Using \eqref{eq:Acal.innerpart.lim-}, \eqref{eq:Acal.innerpart.lim+}, and dominated convergence we find
\begin{equation} \label{eq:lem:Acal.properties.1}
    \lim_{x \to -\infty} \phi \ast H(x) = \inf_{(0,1)} Q_\nu(u) \text{ and } \lim_{x \to +\infty} \phi \ast H(x)= \sup_{(0,1)} Q_\nu(u).
\end{equation}
As $\phi \ast H$ is strictly increasing and continuous, there exist unique points $x_-, x_+ \in \RR$ with 
\begin{align*}
    \min(\supp(\nu)) < \phi \ast H(x_-) &= \min(\supp(\mu)), \\ 
    \max(\supp(\nu)) > \phi \ast H(x_+) &= \max(\supp(\mu)).
\end{align*}
Hence, we have $\Acal F(x_-) = 0$ and $\Acal F(x_+) = 1$, which yields that $\Acal F$ can be identified as the CDF of a measure concentrated on the compact interval $[x_-,x_+]$.
    
To see \ref{it:A.shift}, fix $c,x \in \mathbb R$ and observe that $(\phi \ast F)(x+c) = \int_\mathbb R F(x+c-z)\, \phi(z)dz$.
   Then
    \begin{align*}
        (\mathcal A F)(x + c) 
        &= F_\mu\left(\phi \ast \left(Q_\nu(\phi \ast F)\right)(x+c)\right) 
        \\
        &= F_\mu \left ( \int_\mathbb R Q_\nu((\phi \ast F)(x+c-y)) \phi(y) \, dy \right)
        \\
        &= F_\mu \left( \int_\mathbb R Q_\nu\left(  \int_\mathbb R F(x + c - y - z) \phi(z) \, dz\right) \phi(y) \, dy \right),
    \end{align*}
    which shows the claim, that is $(\mathcal A F(\cdot + c))(x)=(\mathcal A F)(x + c)$.

    To see \ref{it:A.monotone}, note that $F_\mu$ and $Q_\nu$ are both increasing. 
    Furthermore, the convolution with a Gaussian is an order-preserving operator
    , i.e. if $F \leq G$, we have $\phi \ast F \le \phi \ast G$.
    Therefore, all operations involved in the definition of $\mathcal A$ preserve pointwise the order, which yields the claim.

    Finally we show \ref{it:A.symmetry}:
    Let $\mu$, $\nu$ and the distribution with CDF $F$ be symmetric. 
    Write $m_1$ for the median of $F$ and $m_2$ for the median of $\mu$ and $\nu$.
    By symmetry of the distributions we have
    \begin{align}
        \label{eq:Acal.props.4.1}
        F(x) + F(2 m_1 - x) = 1, \\
        \label{eq:Acal.props.4.2}
        \phi \ast F(x) + \phi \ast F(2 m_1 - x) = 1, \\
        \label{eq:Acal.props.4.3}
        F_\mu(x) + F_\mu(2 m_2 - x) = 1, \\
        \label{eq:Acal.props.4.4}
        Q_\nu(u) = 2m_2 - Q_\nu(1 - u),
    \end{align}
    for all $x \in \RR$, $u \in (0,1)$.
    Using these identities we compute
    \begin{align*}
        \Acal F(x) &= 1 - F_\mu \left( 2 m_2 - \phi \ast Q_\nu \left( \phi \ast F \right)(x) \right) \\
        &= 1 - F_\mu \left( 2 m_2 - \int_\RR  Q_\nu \left( \phi \ast F(x - y) \right) \phi(y) \, dy \right) \\
        &= 1 - F_\mu \left( \int_\RR Q_\nu \left( 1 - \phi \ast F(x - y) \right) \phi(y) \, dy \right) \\
        &= 1 - F_\mu \left( \int_\RR  Q_\nu \left( \phi \ast F( 2 m_1 - x + y)\right) \phi(y) \, dy \right) \\
        &= 1 - \Acal F(2 m_1 - x),
    \end{align*}
    where the first equality is due to \eqref{eq:Acal.props.4.3}, the third due to \eqref{eq:Acal.props.4.4}, the fourth follows from \eqref{eq:Acal.props.4.2}, and the last by symmetry of $\phi$.
    We have shown that $\Acal F$ is symmetric as well as that $m_1$ is its median.
\end{proof}

In order to prove regularity of the operator $\Acal$ (cf. Proposition~\ref{mainProposition} below) and Theorem~\ref{thm:regularity}, we  study the following auxiliary functions.

\begin{definition}
 Under Assumption \ref{ass:A1&2}, define, for $Q \in L^0(0,1)$, the functions
\begin{align}
\label{aux_operator_def1}
	S_Q(x) &:= \int_\RR Q_\nu \left( \int_0^1 \Phi(x-Q(y)-z)dy \right)\phi(z)dz,\quad x \in \mathbb R,\\
 T_Q(x,y)& := \int_\RR \phi(x-y-z) Q_\nu'\left( \int_0^1 \Phi(x-Q(w)-z)dw \right) \phi(z) \, dz,\quad x,y \in \mathbb R,\nonumber
\end{align}
and, for $Q,f\in L^\infty(0,1)$, the function 
\begin{equation*}
 S_{Q, f}(\epsilon, x):=S_{Q+\epsilon f}(x),\quad x,\epsilon \in \mathbb R.
\end{equation*}

\end{definition}
The main properties of these functions are proved in Lemma~\ref{auxiliary_lemma}, Lemma~\ref{auxiliary_lemma2} and Lemma~\ref{auxiliary_lemma3}, respectively.
Intuitively speaking, the operator $Q \mapsto S_{Q}$ maps the random variable $Q\in L^0(0,1)$ with law $\alpha$ to $\phi \ast T_{\alpha \ast \gamma, \nu}$, where $T_{\alpha\ast\gamma, \nu}$ is the unique monotone map from $\alpha\ast\gamma$ to $\nu$.
Further, we introduce the operator $(Q,f) \mapsto S_{Q, f}$ to investigate the sensitivity of $S_Q$ under perturbations of $Q$ by $\epsilon f$, where $f \in L^\infty(0,1)$ and $\epsilon > 0$. 
Importantly, the properties of $Q\mapsto T_Q$ which are established in Lemma \ref{auxiliary_lemma3}, will be key to showing Proposition \ref{mainProposition}.

\begin{lemma}
    \label{auxiliary_lemma}
    Under Assumption \ref{ass:A1&2}, let $Q, \tilde Q \in L^0(0,1)$. Then the following statements hold:
    \begin{enumerate}[label = (\roman*)]  
     \item \label{it:auxiliary_lemma.1} $S_Q$ is smooth and strictly increasing with ${\rm range}(S_Q) = \interior (\co(\supp \nu))$.
        \item \label{it:auxiliary_lemma.2} $S_Q'$ is strictly positive, $(x,Q) \mapsto S_Q'(x)$ is continuous on $\RR \times L^0(0,1)$, and 
        \[ 
            \|S_Q' - S_{\tilde Q}'\|_\infty \le L\|Q - \tilde Q\|_\infty.
        \]
        \item \label{it:auxiliary_lemma.3} If $(Q_k)_{k \in \NN} \subseteq L^0(0,1)$ converges in probability to $Q$, then
        \begin{align*}
            \lim_{k \to \infty} S_{Q_k}^{-1} = S_{Q}^{-1}\quad  \text{locally uniformly in }\interior(\co(\supp(\nu))).
        \end{align*}        
    \end{enumerate}
    In particular, for every $R > 0$ there exist $\epsilon_S(R), \delta_S(R) > 0$ such that for all $(x,Q) \in \RR \times L^\infty(0,1)$ with $|x|, \|Q\|_\infty \le R$ we have
    \begin{equation}
        \label{eq:auxiliary_lemma.inparticular}
        \epsilon_S(R) \le S'_Q(x) \le \delta_S(R).
    \end{equation}
\end{lemma}

\begin{remark}[Continuity of $\Acal$]\label{rem:auxiliary_lemma}
    We remark that if $Q \in \QF$, then $S_Q=\phi \ast ( Q_\nu \circ ( \phi \ast F ))$ where $F=Q^{-1} \in \CDF$.
    Moreover, we have $(\Acal F)^{-1} = S_Q^{-1} \circ Q_\mu$ by \eqref{Acal_intro}.
    Let $(Q_k)_{k \in \NN}$ be a sequence in $\QF$ that converges in probability to $Q$.
    Denote by $F_k$ resp.\ $F$ the CDF associated with $Q_k$ resp.\ $Q$.
    Then it follows from Lemma \ref{auxiliary_lemma} \ref{it:auxiliary_lemma.3} that $(\Acal F_k)^{-1} = S_{Q_k}^{-1} \circ Q_\mu \to S_Q^{-1} \circ Q_\mu = (\Acal F)^{-1}$ at all continuity points of $Q_\mu$.
    Since $Q_\mu$ is increasing, this means that this convergence holds in probability.
\end{remark}

\begin{proof}[Proof of Lemma~\ref{auxiliary_lemma}]
    \ref{it:auxiliary_lemma.1}: 
    Since $\nu$ is compactly supported, $Q_\nu$ is bounded and measurable.
    Denote by $H$ the composition of $Q_\nu$ with the right-hand side of \eqref{eq:Acal.innerpart} and note that $S_Q = \phi \ast H$.
    In the proof of Lemma \ref{lem:Acal.properties} \ref{it:A.range},  
    we have shown that $\phi \ast H$ is smooth and strictly increasing.
    Thus, we deduce from \eqref{eq:lem:Acal.properties.1} that the range of $S_Q = \phi \ast H$ is $(\min(\supp(\nu)), \max(\supp(\nu))$, where the boundary is not attained due to Assumption \ref{it:A1}.

        \ref{it:auxiliary_lemma.2}: We have to show that $S_Q'$ is strictly positive.
        To this end, note that
        \begin{equation}
            \label{eq:auxiliary_lemma.H'}
            H'(x) =
            Q_\nu' \left( \int_0^1 \Phi(x-Q(w)) \, dw \right) \int_0^1 \phi(x-Q(w)) \,dw.
        \end{equation}
        Due to \ref{it:A1} we have that $F_\nu$ is continuous and differentiable.
        Using \ref{it:A3} we can apply the inverse function rule and find a constant $L > 0$ with $0 < Q'_\nu \le L$.
        Hence, $\|H'\|_\infty \le L$, which permits us to exchange differentiation with integration in
        \begin{align*}
            S'_Q(x) &= \frac{\partial}{\partial x}
            \int H(x - y) \phi(y) \, dy 
            = \int H'(x-y) \phi(y) \, dy.
        \end{align*}
        In order to conclude that $S'_Q > 0$, we will show $H'(x) > 0$ $dx$-almost everywhere.
        By \eqref{eq:auxiliary_lemma.H'} it remains to establish that the first factor $Q_\nu'(\int_0^1 \Phi(x - Q(w)) \, dw)$ is strictly positive for $dx$-almost every $x$.
        To see this, we define $\alpha$ as the law of $Q(U)$, where $U$ is uniformly distributed on $(0,1)$.
        Let $B_1 \sim \gamma \ast \alpha$.
        Since $\gamma \ast \alpha$ is equivalent to the Lebesgue measure, it suffices to prove 
        \[
            Q_\nu'\left( \int_0^1 \Phi(B_1 - Q(w)) \right) = Q_\nu'\left( \phi \ast F_\alpha (B_1)  \right)> 0\quad \text{a.s.,}
        \]
        where the equality stems from \eqref{eq:Acal.innerpart}.
        Using the fact that $( \phi \ast F_\alpha)(B_1)$ is uniformly distributed on $(0,1)$ and $Q_\nu'(U) > 0$ a.s., yields the claim.

    Let $\tilde Q \in L^0(0,1)$ and write $\tilde H$ for $Q_\nu$ composed with the right-hand side of \eqref{eq:Acal.innerpart} (adequately replacing $Q$ by $\tilde Q$).
    By the previous part we have that $Q_\nu$ is $L$-Lipschitz.
    Hence, as $\Phi$ is 1-Lipschitz we get  for any $x \in \RR$
    \[
        |H(x) - \tilde H(x)| \le L \int_{0}^1 |\Phi(x - Q(w)) - \Phi(x - \tilde Q(w))| \, dw \le L \|Q - \tilde Q\|_\infty.
    \]
    Consequently, we have
    \begin{align}
        \label{eq:auxiliary_lemma.it.3.2}
        |S_Q'(x) - S_{\tilde Q}'(x)| &= |\phi' \ast (H - \tilde H)(x)| \le L \|Q- \tilde Q\|_\infty  \int_\RR |\phi'(x-y)| \, dy \le L\|Q - \tilde Q\|_\infty,
    \end{align}
    where the last inequality follows from the identity $\phi'(x) = - x \phi(x)$ and the fact that the first absolute moment of $\phi$ is less than 1.
    Moreover, if $(Q_k)_{k \in \NN}$ is a sequence in $L^0(0,1)$ that converges in probability to $Q$,
    then it follows from Lemma \ref{lem:aux.ToP} that
    \[
        \lim_{k \rightarrow \infty}  Q_\nu'\left( \int_0^1 \Phi(\cdot - Q_k(w)) \, dw \right)  = \lim_{k \rightarrow \infty}  Q_\nu' \circ F_{\gamma \ast \alpha_k} = Q_\nu' \circ F_{\gamma \ast \alpha}  =
        Q_\nu'\left( \int_0^1 \Phi(\cdot - Q(w)) \, dw \right)
    \]
 in $L^1(\RR; \gamma \ast \alpha)$, where $\alpha_k$ denotes the distribution whose quantile function is $Q_k$. Therefore, we can conclude that 
    \begin{equation*}
        H'_k(x) :=
            Q_\nu' \left( \int_0^1 \Phi(x-Q_k(w)) \, dw \right) \int_0^1 \phi(x-Q_k(w)) \,dw \xrightarrow[k \rightarrow \infty]{} H'(x)
    \end{equation*}
    in $L^1(\RR; \gamma \ast \alpha)$. As $\gamma \ast \alpha$ is equivalent to the Lebesgue measure, we have $\lim_{k \to \infty} S'_{Q_k}(x) = \lim_{k \to \infty} \int H'_k(x-y) \phi(y) \, dy = S'_Q(x)$  by dominated convergence theorem. This yields continuity of $(x,Q) \mapsto S_Q'(x)$ on $\RR \times L^0(0,1)$.

    \ref{it:auxiliary_lemma.3}
    Since $S_Q$ is strictly increasing and continuous, the same holds true for $S_Q^{-1}$.
    For any sequence $(Q_k)_{k \in \NN}$ that converges in probability to $Q$,
    we have by dominated convergence and Lipschitz continuity of $Q_\nu$ that
    \[
        Q_\nu \left( \int_0^1 \Phi(x - Q_k(w)) \, dw \right) \to
        Q_\nu \left( \int_0^1 \Phi(x - Q(w)) \, dw \right),
    \]
    for every $x \in \RR$.
    Using Assumption \ref{it:A3}, we can once more apply dominated convergence to find that $S_{Q_k} \to S_Q$ pointwise.
    Thus, by standard arguments (see, for example, \cite[Proof of Proposition 5 in Chapter 14]{FrGr13}), $(S_{Q_k}^{-1})_{k \in \NN}$ converges pointwise to $S_Q^{-1}$ at all continuity points of $S_Q^{-1}$.
    But, as $S_Q^{-1}$ is continuous, this simply means that $S_{Q_k}^{-1} \to S_Q^{-1}$ pointwise.
    Finally, we note that pointwise convergence of a sequence of increasing functions to a continuous function entails locally uniform convergence.

    Finally, to see the last assertion, we first observe that if $Q_\# \lambda = \tilde Q_\# \lambda$ then also $S_Q = S_{\tilde Q}$.
    Therefore, we can restrict without loss of generality to $Q \in \QF$.
    The set $\{ Q \in \QF : \|Q\|_\infty \le R\}$ is compact in $L^0(0,1)$.
    Indeed, the map $Q \mapsto Q_\# \lambda$ is a homeomorphism between $\QF$ and $\Pcal(\RR)$.
    By Prokhorov's theorem $\Pcal([-R,R])$ is compact, which means that $\{ Q \in \QF : \|Q\|_\infty \le R\}$ is compact in $L^0(0,1)$.
    We conclude \eqref{eq:auxiliary_lemma.inparticular} by \ref{it:auxiliary_lemma.2}.
\end{proof}

\begin{lemma} \label{auxiliary_lemma2}
    Under Assumption \ref{ass:A1&2}, let $Q,f \in L^\infty(0,1)$. Then the following statements hold:
    \begin{enumerate}[label = (\roman*)]
        \item \label{it:aux_lem2.1} $S_{Q,f} \in C^1(\RR^2)$ with $\partial_x S_{Q,f} > 0$.
        \item \label{it:aux_lem2.2} If $(Q_k)_{k \in \NN}$ converges to $Q$ in $L^\infty(0,1)$, then $\partial_\epsilon S_{Q_k,f} \to \partial_\epsilon S_{Q,f}$ in $L^\infty(\RR^2)$.
        \item If $(Q_k)_{k \in \NN}$ converges in probability to $Q$, then $\partial_\epsilon S_{Q_k,f} \to \partial_\epsilon S_{Q,f}$ uniformly on compacts.
    \end{enumerate}
\end{lemma}

\begin{proof}
    Because of Lemma \ref{auxiliary_lemma} \ref{it:auxiliary_lemma.2} it remains to show that $S_{Q, f}$ is continuously differentiable w.r.t.\ $\epsilon$.
    We proceed to show this hand in hand with \ref{it:aux_lem2.2}.
    Note that
    \[
        g_Q(\epsilon,x) := \int_0^1 \Phi(x - Q(w) - \epsilon f(w)) \, dw
    \]
    is continuously differentiable with uniformly bounded derivative.
    Due to the Assumptions \ref{it:A1} and \ref{it:A3}, we have that $Q_\nu$ is differentiable on $\co(\supp(\nu))$ with bounded derivative.
    Therefore, we can exchange differentiation with integration and get
    \begin{equation}
        \label{eq:auxiliary_lemma.it.4.1}
         \frac{\partial}{\partial \epsilon} S_{Q,f}(\epsilon, x) =
        \int_\RR  \left( Q_\nu' \circ g_Q(\epsilon,y) \right) \frac{\partial}{\partial\epsilon}g_Q(\epsilon,y) \phi(x - y) \, dy 
    \end{equation}
    Clearly, this function is continuous in $x$.
    To see continuity in $\epsilon \in \RR$, let $B_1 \sim \mathcal N(0,1)$ be independent of $U \sim \text{Unif}_{[0,1]}$, and write $B_{\tilde \epsilon, \tilde Q} := B_1 + \tilde Q(U) + \tilde  \epsilon f(U)$ for $\tilde Q \in L^\infty(0,1)$.
    In fact, $g_{\tilde Q}(\tilde  \epsilon,\cdot)$ is the CDF of $B_{\tilde  \epsilon,\tilde Q}$ and $(\tilde  \epsilon, \tilde Q) \mapsto Law(B_{\tilde  \epsilon, \tilde Q})$ is continuous with domain $\RR \times L^\infty(0,1)$ and codomain $\Pcal(\RR)$ equipped with the total variation norm.
    Applying Lemma \ref{lem:aux.ToP} yields
    \[
        \lim_{(\tilde \epsilon, \tilde Q) \to (\epsilon, Q)} Q_\nu' (g_{\tilde Q}(\tilde \epsilon,\cdot)) = 
        Q_\nu' (g_Q (\epsilon,\cdot) )\quad \text{in } L^1(\RR; Law(B_{\epsilon, Q})).
    \]
    As $Law(B_{\epsilon,Q})$ is equivalent to the Lebesgue measure, this convergence holds also in probability w.r.t.\ the Lebesgue measure.
    Further, we have the upper bound
    \begin{align*}
        |\partial_\epsilon g_Q(\epsilon,x)| 
        &\le \int_0^1 | f(w) | \phi(x - Q(w) - \epsilon f(w)) \, dw
        \\
        &\le \|f\|_\infty \phi(x) \int_0^1 e^{x (Q(w)  + \epsilon f(w)) } \, dw \le \|f\|_\infty \phi(x) e^{|x| \| Q + \epsilon f\|_\infty }.        
    \end{align*}
    By dominated convergence, we have that
    \[
        \lim_{(\tilde\epsilon,\tilde Q) \to (\epsilon,Q)} \partial_\epsilon g_{\tilde Q}(\tilde \epsilon,\cdot) =
        \partial_\epsilon g_Q(\epsilon,\cdot) \quad \text{in }L^1(\RR).
    \]    
    By Assumption \ref{it:A3} we have that $Q'_\nu$ is absolutely bounded, hence, again by dominated convergence
    \begin{equation}
        \label{eq:aux_lem2.1}
        \lim_{(\tilde\epsilon,\tilde Q) \to (\epsilon,Q)} \frac{\partial}{\partial\epsilon} (Q_\nu \circ g_{\tilde Q} (\tilde \epsilon,\cdot) ) =  \frac{\partial}{\partial\epsilon} (Q_\nu \circ g_{Q} (\epsilon,\cdot) ) \quad \text{in }L^1(\RR).
    \end{equation}
    By \eqref{eq:auxiliary_lemma.it.4.1}, for every $x\in\RR$ we get
    \begin{align*}
        \left| \frac{\partial}{\partial \epsilon} S_{\tilde Q, f}(\tilde \epsilon, x) - \frac{\partial}{\partial \epsilon} S_{Q, f}(\epsilon, x) \right| &\le 
        \int_\RR \left| \frac{\partial}{\partial \epsilon} \left( Q_\nu \circ g_{\tilde Q}(\tilde \epsilon, y) - Q_\nu \circ g_{Q}(\epsilon, y) \right) \right| \, \phi(x-y) \, dy \\
        &\le \left\| \frac{\partial}{\partial \epsilon} \left( Q_\nu \circ g_{\tilde Q}(\tilde \epsilon, \cdot) - Q_\nu \circ g_{Q}(\epsilon, \cdot) \right) \right\|_{L^1(\RR)},
    \end{align*}
    where we used $\|\phi\|_\infty \le 1$ in the second inequality.
    Due to \eqref{eq:aux_lem2.1}, the right-hand side vanishes for $(\tilde \epsilon, \tilde Q) \to (\epsilon,Q)$ uniformly in $x$, which concludes the proof.
\end{proof}

\begin{lemma} \label{auxiliary_lemma3}
    Under Assumption \ref{ass:A1&2}, let $Q \in L^0(0,1)$. Then the following statements hold:
    \begin{enumerate}[resume, label = (\roman*)]
        \item \label{it:auxiliary_lemma3.1} $T_Q$ is well-defined, bounded, continuous and strictly positive.
        \item \label{it:auxiliary_lemma3.2} If $(Q_k)_{k \in \NN} \subseteq L^0(0,1)$ converges in probability to $Q$, then
        \[ \lim_{k \to \infty} T_{Q_k} = T_{Q} \quad \text{uniformly.} \]
    \end{enumerate}
    In particular, for every $R > 0$ there are $\epsilon_T(R), \delta_T(R) > 0$ such that for all $x,y\in \RR$, $Q \in L^\infty(0,1)$ with $|x|, |y|, \|Q\|_\infty \le R$ we have
    \begin{equation}
        \label{eq:auxiliary_lemma3.inparticular}
        \epsilon_T(R) \le T_Q(x,y) \le \delta_T(R).  
    \end{equation}
\end{lemma}

\begin{proof}
    By Assumption \ref{it:A3}, we have that $Q'_\nu$ is bounded by $L$, therefore $T_Q$ is well-defined and bounded by $L$.
    Using dominated convergence, it is easy to see that $T_Q$ is continuous.
    Finally, we have $T_Q > 0$ as convolution of a strictly positive function with a Gaussian.

    We proceed to show \ref{it:auxiliary_lemma3.2}.    
    As in the proof of Lemma \ref{auxiliary_lemma2} one can show that
    \[
        Q \mapsto Q_\nu'\left( \int_0^1 \Phi(\cdot - Q(w)) \, dw \right)
    \]
    is continuous from $L^0(0,1)$ to $L^0(\RR)$, when domain and codomain are both equipped with convergence in probability.
    Since $Q_\nu'$ is bounded, dominated convergence yields $(L^0(0,1),L^2(\RR))$-continuity of
    \[
        Q \mapsto f_Q(\cdot) := \phi(\cdot) Q_\nu'\left( \int_0^1 \Phi(\cdot -Q(w)) \, dw \right).
    \]
    Hence, by the Cauchy-Schwartz inequality we find that
    \begin{align*}
        |T_Q(x,y) - T_{Q_k}(x,y)| &\le \left( \int_\RR \phi^2(x - y - z) \, dz \right)^\frac12 \|f_Q - f_{Q_k} \|_{L^2(\RR)} 
        \\
        &\le
        \|f_Q - f_{Q_k} \|_{L^2(\RR)},
    \end{align*}
    from where we conclude that $T_{Q_k} \to T_Q$ uniformly.    

    To see \eqref{eq:auxiliary_lemma3.inparticular}, note that by the same reasoning given as in the final part of the proof of Lemma \ref{auxiliary_lemma}, we can restrict to the $L^0(0,1)$-compact set $\{ Q \in \QF : \|Q\|_\infty \le R\}$.
    Therefore, \eqref{eq:auxiliary_lemma3.inparticular} follows from continuity and strict positivity of $(x,y,Q) \mapsto T_Q(x,y)$ on the domain $\RR^2 \times L^0(0,1)$.
\end{proof}

\subsection{Regularity} \label{ssec:regularity}
To establish convergence of the fixed-point algorithm, we reinterpret $\Acal$, cf.\ \eqref{Acal_intro}:
Instead of working with CDFs, we consider an operator that acts on quantile functions.
This point of view allows us to define an operator $\Gcal \colon L^\infty(0,1) \to L^\infty(0,1)$ that essentially takes the role of $\Acal$.
In the main result of this section, we exploit the linear structure of $L^\infty(0,1)$ in order to study the regularity of $\Gcal$ (and thereby implicitly of $\Acal$).
This step plays a crucial role in the proof of our main results.

\begin{definition}
	Under Assumption \ref{ass:A1&2}, we define the operator $\Gcal \colon L^\infty(0,1) \to L^\infty(0,1)$
    \[
         Q \mapsto S_Q^{-1} \circ Q_\mu,
    \]
    where $Q \mapsto S_Q$ is given as in  \eqref{aux_operator_def1}.
\end{definition} 

We proceed by fleshing out the connection between $\Gcal$ and $\Acal$, cf.\ especially Proposition \ref{mainProposition} \ref{it:mainProposition.2}.

\begin{proposition}
    \label{mainProposition}
    Under Assumption \ref{ass:A1&2}, $\Gcal$ has the following properties:
    \begin{enumerate}[label = (\roman*)]
        \item \label{it:mainProposition.1}  For every $Q \in L^\infty(0,1)$, $u \mapsto \Gcal Q(u) =: \Gcal_u (Q)$ is right-continuous.
        \item \label{it:mainProposition.2}  For every $F \in {\rm CDF}$ with quantile function $Q \in L^\infty(0,1)$, we have $(\Acal F)^{-1} = \Gcal Q$;
        \item \label{it:mainProposition.3} The operator $\Gcal$ is Fréchet differentiable with derivative $D_Q \Gcal$;
        \item \label{it:mainProposition.4} For every $Q \in L^\infty(0,1)$ the operator norm of the Fréchet derivative $D_Q\Gcal$ is bounded by 1 and
        \[
            \| D_Q \Gcal(f) \|_\infty = \|f \|_\infty \iff f \in {\rm Span}(\mathbf 1).
        \]
    \end{enumerate}
\end{proposition}

\begin{proof}
Before starting with proving the statements, we remark that by Lemma \ref{auxiliary_lemma}, for any $Q \in L^\infty(0,1)$, the inverse $S_Q^{-1}$ has domain $\co(\interior(\supp(\nu))) \supseteq \supp(\mu)$.
Consequently, $\Gcal Q = S_Q^{-1} \circ Q_\mu$ is well-defined and \ref{it:mainProposition.1}  follows from continuity of $S^{-1}_Q$ and right-continuity of $Q_\mu$.
Further, $\Gcal$ is a continuous operator, since the map $Q \mapsto S_Q^{-1}|_{\supp(\mu)}$ is by Lemma \ref{auxiliary_lemma} \ref{it:auxiliary_lemma.3} continuous with domain $L^\infty(0,1)$ and codomain $L^\infty(\supp(\mu))$.

We proceed to show \ref{it:mainProposition.2}:
Let $F$ be a CDF with quantile function $Q \in L^\infty(0,1)$. 
For any $u \in (0,1)$ we have
\begin{align*}
    (\Acal F)^{-1}(u) & = \inf \{ x \in \mathbb R : \mathcal AF(x) \ge u \} \\
    &= \inf \{ x \in \RR : F_\mu(S_{ Q}(x)) \ge u \} \\
    &= \inf \{ S_{ Q}^{-1}(x) : x \in \interior(\co(\supp(\nu))), F_\mu(x) \ge u \}\\
    &= S_{ Q}^{-1} \circ Q_\mu(u) = \Gcal Q(u).
\end{align*}
To see \ref{it:mainProposition.3} , let $Q, f \in L^\infty(0,1)$, $u \in (0,1)$  and $\epsilon\in\RR$.
We have
\[
    Q_\mu(u) = S_{Q + \epsilon f} \circ \Gcal_u(Q + \epsilon f) = S_{Q,f}(\epsilon,\Gcal_u(Q + \epsilon f)).
\]
By Lemma \ref{auxiliary_lemma2}, $S_{Q,f}$ is $C^1(\RR^2)$ with $\partial_x S_{Q,f}(\epsilon,x) > 0$ for all $(\epsilon,x) \in \RR^2$.
Therefore, applying the implicit function theorem yields that $\Gcal_u(Q + \epsilon f)$ is differentiable w.r.t.\ $\epsilon$.
We find
\begin{equation}
    \label{eq:mainProposition1}
    0 = \frac{\partial}{\partial \epsilon}\Bigr|_{\epsilon = 0} (S_{Q,f}(\epsilon,\Gcal_u(Q+\epsilon f)) = 
    \partial_\epsilon S_{Q,f}(0,\Gcal_u(Q))  + 
    \partial_x S_{Q,f}(0,\Gcal_u(Q)) D_{Q} \Gcal_u(f).
\end{equation}
Note that $S_{Q,f}(0,x) = S_Q(x)$ and thus $\partial_x S_{Q,f}(0,x) = S_Q'(x) > 0$ by Lemma \ref{auxiliary_lemma} \ref{it:auxiliary_lemma.2}.
We can rearrange the terms in \eqref{eq:mainProposition1} to obtain the following representation of $D_Q \Gcal_u(f)$
\begin{equation}
    \label{eq:mainProposition2}
    D_Q\Gcal_u(f) = - \frac{\partial_\epsilon S_{Q,f}(0,\Gcal_u(Q))}{S_Q'(\Gcal_u(Q))}.
\end{equation}
This allows us to define the map
\begin{equation}
    \label{eq:mainProposition3}
    D_Q\Gcal \colon L^\infty(0,1) \to L^\infty(0,1) \colon f \mapsto \left( u \mapsto D_Q \Gcal_u(f) \right),
\end{equation}
which is by construction the Gâteaux derivative of $\Gcal$ at $Q$ in direction $f$.
For $v \in (0,1)$, we let
\begin{equation}
    \label{eq:mainProposition4}
    D_Q\Gcal_u(v) := \frac{T_Q(\Gcal_u Q, Q(v))}{S'_Q(\Gcal_u Q)},
\end{equation}
where we recall that
\[
    S'_Q(x) = \int_0^1 T_Q(x, Q(v)) \, dv,
\]
and $T_Q$ is a strictly positive map by Lemma \ref{auxiliary_lemma3} \ref{it:auxiliary_lemma3.1}.
Thus, we observe that $S'_Q$ acts as a positive normalization constant and $D_Q \Gcal_u$ is a probability density.
By construction, we obtain for all $f \in L^\infty(0,1)$
\begin{equation}
\label{eq:mainProposition5}
        D_Q\Gcal_u(f) = \int_0^1 f(v) D_Q \Gcal_u(v) \, dv.
\end{equation}
Consequently, $D_Q \Gcal_u$ is a linear operator with operator norm 1, and in particular continuous.

To conclude that $\Gcal$ is (continuously) Fréchet differentiable, it suffices to show continuity of $Q \mapsto D_Q \Gcal$ from $L^\infty(0,1)$ to $(L^\infty(0,1))'$.
To this end, let $(Q_k)_{k \in \NN}$ be a sequence that converges to $Q$ in $L^\infty(0,1)$. In particular, we have
\begin{align}
    \sup_{\| f \|_\infty = 1} \| D_{Q_k} \Gcal(f)  - D_Q \Gcal (f) \|_\infty 
    &\leq \sup_{\| f \|_\infty = 1} \lambda\text{-}\esssup_{u \in (0,1)}\int_0^1 |f(v)| |D_{Q_k} \Gcal_u(v) - D_{Q} \Gcal_u(v)| \, dv \nonumber \\
    \nonumber
    &\leq \int_0^1 \lambda\text{-}\esssup_{u \in (0,1)} |D_{Q_k} \Gcal_u(v)-D_Q \Gcal_u(v)| \, dv \\ \label{eq:mainPropositionFrechetcont}
    &\le \lambda^2\text{-}\esssup_{u,v \in (0,1)} |D_{Q_k} \Gcal_u(v)-D_Q \Gcal_u(v)|.
\end{align}
As $\Gcal$ is continuous, there is $R \in \RR$ such that $\sup_{k \in \NN} \| \Gcal Q_k \|_\infty \le R$ and $\sup_{k \in \NN} \|Q_k\|_\infty \le R$.
By Lemma \ref{auxiliary_lemma} there is $\epsilon_S(R) > 0$ such that $\inf_{k \in \NN} S'_{Q_k}(x) \ge \epsilon_S(R)$ for all $x \in [-R,R]$.
Moreover, Lemma \ref{auxiliary_lemma} \ref{it:auxiliary_lemma.2} and Lemma \ref{auxiliary_lemma3} \ref{it:auxiliary_lemma3.2} establish that $S'_{Q_k} \to S'_Q$ uniformly and $T_{Q_k} \to T_Q$ locally uniformly, respectively.
Combining these observations yields that
\begin{equation}
    \label{eq:mainPropositionFrechet1}
    \frac{T_{Q_k}(x,y)}{S'_{Q_k}(x)} \to \frac{T_Q(x,y)}{S'_Q(x)} \quad \text{locally uniformly for }(x,y) \in \RR^2.
\end{equation}
By Lemma \ref{auxiliary_lemma} \ref{it:auxiliary_lemma.1} and Lemma \ref{auxiliary_lemma3} \ref{it:auxiliary_lemma3.1} we have continuity of $S_Q'$ and $T_Q$, respectively, from where we deduce continuity of $(x,y) \mapsto T_Q(x,y) / S'_Q(x)$.
We find by simple computations that
\begin{multline} \label{eq:mainPropositionFrechet}
    \left| D_{Q_k} \Gcal_u(v) - D_Q \Gcal_u(v) \right|
    \\
    \le
    \left| \frac{T_{Q_k}(\Gcal_u Q_k, Q_k(v))}{S_{Q_k}'(\Gcal_u Q_k)} - \frac{T_{Q}(\Gcal_u Q_k, Q_k(v))}{S_{Q}'(\Gcal_u Q_k)}\right| 
    + \left| \frac{T_{Q}(\Gcal_u Q_k, Q_k(v))}{S_{Q}'(\Gcal_u Q_k)} - \frac{T_{Q}(\Gcal_u Q, Q(v))}{S_{Q}'(\Gcal_u Q)} \right|.
\end{multline}
Since $Q_k \to Q$ in $L^\infty(0,1)$ and $\Gcal$ is continuous, this means that
\[
    (\Gcal_u Q_k, Q_k(v)) \to (\Gcal_u Q, Q(v)) \quad \text{uniformly, for $\lambda^2$-almost every }(u,v) \in (0,1)^2.
\]
The first term in \eqref{eq:mainPropositionFrechet} vanishes by \eqref{eq:mainPropositionFrechet1}, as $k \to \infty$ for $\lambda^2$-almost every $(u,v)$, while the second term in \eqref{eq:mainPropositionFrechet} vanishes by uniform continuity of the map $(x,y) \mapsto T_Q(x,y) / S'_Q(x)$ on $[-R,R]^2$, as $k \to \infty$ for $\lambda^2$-almost every $(u,v) \in (0,1)^2$.
We have shown by \eqref{eq:mainPropositionFrechetcont} that $D_{Q_k} \Gcal \to D_Q \Gcal$, which means that $Q \mapsto D_Q \Gcal$ is continuous.

Finally, we prove \ref{it:mainProposition.4}:
It follows from \eqref{eq:mainProposition4} that $D_Q \Gcal(\textbf{1}) = \textbf{1}$.
Furthermore, we have by Lemma \ref{auxiliary_lemma} and Lemma \ref{auxiliary_lemma3} that there are $\delta_S, \epsilon_T > 0$ with $S'_Q(x) \le \delta_S$ and $T_Q(x,y) > \epsilon_T$ for all $x,y \in \RR$ with $|x|,|y| \le \|Q\|_\infty$.
Write $\tilde \epsilon := \epsilon_T / \delta_S$ and observe that for $\lambda^2$-almost every $(u,v) \in (0,1)^2$ we have $D_Q \Gcal_u(v) \ge \tilde \epsilon > 0$.
Let $f \in L^\infty(0,1)$ with $f \not \in {\rm Span}(\mathbf 1)$.
Since we can assume $\|f\|_\infty = 1$ w.l.o.g., then $\int_0^1 1 - |f(v)| \, dv =: \tilde \delta > 0$, so that
\begin{align*}
    |D_Q \Gcal_u(f)| &\le \int_0^1 |f(v)| D_Q \Gcal_u(v) \, dv  \\
    &= \int_0^1 |f(v)| (D_Q \Gcal_u(v) - \tilde\epsilon ) \, dv + \tilde \epsilon \int_0^1 |f(v)| \, dv \\
    &\le 1 - \tilde \epsilon + \tilde \epsilon \int_0^1 |f(v)| \, dv = 1 - \tilde \epsilon \tilde \delta.
\end{align*}
Hence, $\|D_Q \Gcal(f)\|_\infty \le 1 - \tilde \epsilon \tilde \delta < 1$, which completes the proof. \qedhere
\end{proof}

Finally, we need the following remarks to proceed with the proof of our main results.  
\begin{remark} [Gradient of $\Gcal$]
    As we showed in the proof of Proposition \ref{mainProposition}\ref{it:mainProposition.3}, the Frèchet differential of $\Gcal$ at $Q \in L^\infty(0,1)$ can be identified with the map
    \[
        (0,1)^2 \ni (u,v) \mapsto D_Q\Gcal_u(v) = \frac{T_Q(\Gcal_u Q, Q(v))}{S'_Q(\Gcal_u Q)}, \quad S'_Q(\Gcal_u Q) = \int_0^1 T_Q(x, Q(v)) \, dv,
    \]
    cf.\ \eqref{eq:mainProposition4}.
    More precisely, $(u,v) \mapsto D_Q\Gcal_u(v)$ can be identified with the gradient of the operator $\Gcal$ since
    \begin{equation*}
        D_Q\Gcal_u(f) = \int_0^1 f(v) D_Q \Gcal_u(v) \, dv,
    \end{equation*}
    for any $f \in L^\infty(0,1)$ and $u \in (0,1)$, cf.\ \eqref{eq:mainProposition5}.
\end{remark}
\begin{remark}
\label{rmk:elliptic}
    It follows from the proof of Proposition \ref{mainProposition} that the gradient of $\Gcal$
    is bounded and bounded away from zero on $\overline{B_\infty(0, R)} := \{ Q \in L^\infty(0,1) : \|Q\|_\infty \le R\}$ where $R > 0$.
    This means, there exist $\epsilon(R), \delta(R) > 0$ such that for $\lambda^2$-almost every $(u,v) \in (0,1)^2$
    \begin{equation}
        \label{eq:elliptic}
        \epsilon(R) \leq D_Q \Gcal_u(v) \leq \delta(R), \quad \text{for any } Q \in \overline{B_\infty(0, R)}.
    \end{equation}
    In fact, $Q \in L^\infty(0,1)$ with $\|Q\|_\infty \le R$ implies by non-expansiveness (cf.\ Theorem \ref{thm:regularity} \ref{it:L_inf_Thm1}) that $\| \Gcal Q - \Gcal 0\|_\infty \le R$.
    Thus, there exists $\tilde R=\tilde R(R)>0$ such that $\|Q\|_\infty, \|\Gcal Q\|_\infty \le \tilde R$ for all $Q \in L^\infty(0,1)$ with $\|Q\|_\infty \le R$.    
    Therefore, by Lemma \ref{auxiliary_lemma}, there exist $\epsilon_S(\tilde R), \delta_S(\tilde R)>0$ such that $S'_{Q}(\Gcal_u Q) \in [\epsilon_S(\tilde R), \delta_S(\tilde R)]$ for any $\lambda$-almost every $u \in (0,1)$.
    Similarly, by Lemma \ref{auxiliary_lemma3} \ref{it:auxiliary_lemma3.2}, there exist $\epsilon_T(\tilde R), \delta_T(\tilde R) > 0$ such that $T_Q(\Gcal_u Q, Q(v)) \in [\epsilon_T(\tilde R), \delta_T(\tilde R)]$ for $\lambda^2$-almost every $(u, v) \in (0,1)^2$. 
    Therefore, \eqref{eq:elliptic} holds with 
    \[ \epsilon(R):=\epsilon_T(\tilde R)/\delta_S(\tilde R) \quad \text{and} \quad \delta(R):=\delta_T(\tilde R)/\epsilon_S(\tilde R). \] 
\end{remark}

\subsection{Proof of the main results}\label{sect:main_proofs}
This section is devoted to the proofs of Theorem~\ref{thm:regularity}, Theorem~\ref{thm:conv}, and to show linear convergence under further assumptions.

\begin{proof}[Proof of Theorem~\ref{thm:regularity}]
    To see  \ref{it:L_inf_Thm1}, let $F, G \in \CDFc$. 
    We have
    \begin{align*}
        \Wcal_\infty  (\Acal F, \Acal G) 
         &= |\!| \Gcal(F^{-1})- \Gcal(G^{-1})|\!|_\infty
         \\
         &= |\!| D_{Q_t} \Gcal (F^{-1}-G^{-1})|\!|_\infty
        \leq |\!| D_{Q_t} \Gcal |\!|_\infty |\!| F^{-1}-G^{-1}|\!|_\infty,
    \end{align*}
    where the first equality is due to Proposition \ref{mainProposition} \ref{it:mainProposition.2} , and
    the latter follows from Taylor's theorem \cite[Section 5.6]{Ca64} with $Q_t = tF^{-1}+(1-t)G^{-1}$, for some $t \in [0,1]$.
    Since $|\!| D_{Q_t} \Gcal|\!|_\infty \leq 1$ by Proposition \ref{mainProposition} \ref{it:mainProposition.3}, $\Wcal_\infty  (\Acal F, \Acal G) \leq \Wcal_\infty  (F, G)$. Furthermore, $\Wcal_\infty  (\Acal F, \Acal G) = \Wcal_\infty  (F, G)$ if and only if $|\!| D_{Q_t} \Gcal (F^{-1}-G^{-1})|\!|_\infty = |\!| F^{-1}-G^{-1}|\!|_\infty$, which is equivalent to $F^{-1}-G^{-1} \in {\rm Span}(\mathbf 1)$, by Proposition \ref{mainProposition} \ref{it:mainProposition.4}.
    
    We derive \ref{it:L_inf_Thm2} from \ref{it:L_inf_Thm1}: Indeed, if $F, G \in \CDFc$ are fixed points of $\Acal$, then $\Wcal_\infty (\Acal F, \Acal G) = \Wcal_\infty (F, G)$. Therefore, $F^{-1} = G^{-1}+ c \mathbf{1}$, for some $c \in \RR$.

     To prove \ref{it:L_inf_Thm3}, note that $\Acal$ admits a fixed-point if and only if there exists a \ssbm from $\mu$ to $\nu$ by Theorem~\ref{char_s2Bm}. Moreover, a \ssbm from $\mu$ to $\nu$ exists if and only if $(\mu, \nu)$ is irreducible by Theorem~\ref{thm.sbm}.
\end{proof}

\begin{proof}[Proof of Corollary~\ref{cor:non-exp}]
To emphasize the dependence of the operators $Q \mapsto S_Q$ and $Q \mapsto \Gcal Q$ on $\mu$ and $\nu$, we introduce the following notations 
\begin{align*}
       	 S_Q^{\nu}(x) &:= \int_\RR Q_\nu \left( \int_0^1 \Phi(x-Q(y)-z)dy \right)\phi(z)dz,\\
     \Gcal_{\mu, \nu}Q &:= (S_Q^\nu)^{-1} \circ Q_\mu.
\end{align*}
Even though we are only assuming that $\nu$ is compactly supported, we can prove that the operator $Q \mapsto S_Q$ is well-defined, smooth and strictly increasing, with ${\rm range}(S_Q^{ \nu}) = \interior (\co(\supp \nu))$ for any $Q \in L^0(0,1)$, following the proof of Lemma \ref{auxiliary_lemma}. Now, let us invoke Lemma \ref{lem:approx.corol} to find a sequence of pairs $(\mu_n, \nu_n)$, $n \in \NN,$ that satisfy Assumption~\ref{ass:A1&2} and 
\begin{equation}
\label{approxConvCor}
\lim_{n \to \infty} \mathcal W_1(\mu_n,\mu) = 0 = \lim_{n \to \infty} \mathcal W_1(\nu_n,\nu).
\end{equation}
Notably, \eqref{approxConvCor} implies weak convergence of $\nu_n$ to $\nu$. Therefore, $Q_{\nu_n}$ converges to $Q_{\nu}$ pointwise for all continuity points of $Q_\nu$, which entails pointwise converges of $S_Q^{\nu_n}$ to $S_Q^{\nu}$ by the dominated convergence theorem, for any $Q \in L^0(0,1)$. As we already showed in the proof of Lemma \ref{auxiliary_lemma} \ref{it:auxiliary_lemma.3}, this implies locally uniform convergence of $(S_Q^{\nu_n})^{-1}$ to $(S_Q^{\nu})^{-1}$. Hence, we have
\begin{equation*}
    (\Gcal_{\mu, \nu} Q)(u) = \lim_{n \rightarrow \infty} ((S_{Q_1}^{\nu_n})^{-1} \circ Q_{\mu_n})(u) = \lim_{n \rightarrow \infty} (\Gcal_{\mu_n, \nu_n} Q)(u),\quad \text{for $\lambda$-a.e. $u \in (0,1)$}.
\end{equation*}
As $(\mu_n, \nu_n)$ satisfies Assumption \ref{ass:A1&2} for any $n \in \NN$, we can invoke Theorem \ref{thm:regularity} \ref{it:L_inf_Thm1} to get
\begin{align*}
    \Wcal_\infty(\Acal F_1, \Acal F_2) & = \| \Gcal_{\mu, \nu} (F_1^{-1}) - \Gcal_{\mu, \nu} (F_2^{-1}) \|_\infty \\  & =  \left \| \lim_{n \rightarrow \infty} \left(\Gcal_{\mu_n, \nu_n} (F_1^{-1}) - \Gcal_{\mu_n, \nu_n} (F_2^{-1}) \right)\right \|_\infty \\ & \leq \liminf_{n \rightarrow \infty} \| \Gcal_{\mu_n, \nu_n} (F_1^{-1}) - \Gcal_{\mu_n, \nu_n} (F_2^{-1}) \|_\infty \\
     & \leq \| F_1^{-1} - F_2^{-1} \|_\infty = \Wcal_\infty(F_1,F_2),
\end{align*}
for any $F_1, F_2 \in \CDF_c$.
\end{proof}

\begin{proof}[Proof of Theorem~\ref{thm:conv}]
    Recall that, due to Proposition \ref{mainProposition} \ref{it:mainProposition.2}, the fixed-point operator $\Acal$ (which acts on CDFs) is equivalent to the operator $\Gcal$ (which acts on quantile functions).
    Therefore, we proceed to work with $\Gcal$.
    Furthermore, we may w.l.o.g.\ assume that $F \in \CDFc$ with quantile function $Q \in L^\infty(0,1)$.
    We write $Q_k := \Gcal^k Q$ for $k \in \NN$.
    Since we assume, in addition to Assumption~\ref{ass:A1&2}, that $(\mu,\nu)$ is irreducible, we know by Theorem \ref{thm:regularity} that there exists a unique (up to translation) solution $Q^\ast$ with $\Gcal Q^\ast = Q^\ast$.
    Put differently, we have 
    \[ \Lcal^\ast := \{ Q \in \QFc : \Gcal Q = Q \} = \{ Q^\ast + c : c\in \RR\}. \]
    \ref{it:conv_1}: We write $Q^\ast_k$ for an element in $\Lcal^\ast$ with minimal $L^\infty$-distance to $Q_k$.
    Indeed, such an element exists because $Q \mapsto \|Q - Q_k\|_\infty$ is lower semicontinuous on $L^0(0,1)$ while $\{ Q \in \Lcal^*: \|Q\|_\infty \le R \}$ is compact in $L^0(0,1)$; cf.\ reasoning in final part of the proof of Lemma \ref{auxiliary_lemma}.
    We have to show that there is $q \in (0,1)$ with
    \begin{equation}
        \label{eq:thm.conv1.toshow}
        \|Q^\ast_{k + 1} - Q_{k + 1}\|_\infty \le q \|Q^\ast_k - Q_k\|_\infty,\quad k \in \NN.
    \end{equation}
    Observe that, by Proposition \ref{mainProposition}, $\Gcal$ is continuously Fréchet differentiable.
    Therefore, we can invoke Taylor's theorem \cite[Section 5.6]{Ca64} to find, for every $k \in \NN$, some $t \in (0,1)$ and $\tilde Q_k = t Q_k^\ast + (1-t)Q_k$ such that
    \begin{equation}
        \label{eq:taylor}
         Q_{k}^\ast - Q_{k + 1} =  \Gcal(Q_{k}^\ast) - \Gcal(Q_{k}) = D_{\tilde Q_k} \Gcal(Q_{k}^\ast - Q_k).
    \end{equation}
    Furthermore, by $L^\infty$-non-expansiveness of $\Gcal$ (see Theorem \ref{thm:regularity} \ref{it:L_inf_Thm1}), the sequence $(\|Q^\ast_1 - Q_k\|_\infty)_{k \in \NN}$ is non-increasing since
    \[
     \|Q^\ast_1 - Q_{k+1}\|_\infty = \|\Gcal(Q^\ast_1) - \Gcal(Q_{k})\|_\infty \leq  \|Q^\ast_1 - Q_{k}\|_\infty,
    \]
     for every $k \in \NN$. Thus, $\sup_{k \in \NN} \|Q^\ast_1 - Q_k\|_\infty <\infty$
    as well as $\sup_{k \in \NN} \|Q^\ast_k - Q_k\|_\infty  \leq \sup_{k \in \NN} \|Q^\ast_1 - Q_k\|_\infty <\infty$, where the first inequality follows from the definition of $Q_k^*$, and $\sup_{k \in \NN} \|Q^\ast_1 - Q_k^*\|_\infty  \leq 2 \sup_{k \in \NN} \|Q^\ast_1 - Q_k\|_\infty <\infty$.
    Consequently, there is $R > 0$ such that $\| \tilde Q_k \|_\infty \le R$ for all $k \in \NN$.
    By Remark \ref{rmk:elliptic}, there are $\epsilon = \epsilon(R) > 0$ and $\delta = \delta(R) > 0$ such that, for all $k \in \mathbb N$,
    \begin{equation}
        \label{eq:densitybounds}
        \epsilon \le D_{\tilde Q_k}\Gcal_u(v) \le \delta \quad \text{for $\lambda^2$-a.e. }(u,v) \in (0,1)^2.
    \end{equation}

    In order to find a factor $q \in (0,1)$ satisfying \eqref{eq:thm.conv1.toshow}, we fix $k \in \mathbb N$ and pick any $y \in [0, \min(\frac{2}{\delta},1)]$.
    We proceed with a case distinction.

    \emph{Case 1}: Assume that $|\int_0^1 Q_k^\ast(v) - Q_k(v) \, dv| \ge (1 - y) \|Q_k^\ast - Q_k\|_\infty$. 
    Writing $s$ for the sign of $\int_0^1 Q_k^\ast(v) - Q_k(v) \, dv$, we have
    \begin{align} \nonumber
        y \| Q_k^\ast - Q_k \|_\infty &\ge \|Q_k^\ast - Q_k\|_\infty - s \int_0^1 Q_k^\ast(v) - Q_k(v) \, dv \\ \nonumber
        &\ge \int_0^1 |s(Q_k^\ast(v) - Q_k(v))| - s(Q_k^\ast(v) - Q_k(v)) \, dv
        \\ \label{eq:lower_bound}
        &= -2 \int_0^1 \min\big( s(Q_k^\ast(v) - Q_k(v)), 0 \big) \, dv.
    \end{align}
    We use this to find for $\lambda$-a.e.\ $u \in (0,1)$ the lower bound
    \begin{align} \nonumber
        s(Q_k^\ast(u) - Q_{k + 1}(u)) &= s \int D_{\tilde Q_k} \mathcal G_u(v) (Q_k^\ast(v) - Q_k(v))\, dv \\ \nonumber
        &\ge \int D_{\tilde Q_k} \mathcal G_u(v) \min\big( s(Q_k^\ast(v) - Q_k(v)), 0 \big) \, dv
        \\ \label{eq:lower_bound.2}
        &\ge \delta \int D_{\tilde Q_k} \min\big( s(Q_k^\ast(v) - Q_k(v)), 0 \big) \, dv \ge - \frac{\delta y}{2} \|Q_k^\ast - Q_k\|_\infty,
    \end{align}
    where the equality follows from \eqref{eq:taylor}, the second inequality from \eqref{eq:densitybounds} and the final inequality from \eqref{eq:lower_bound}.
    Hence, using \eqref{eq:lower_bound.2} and non-expansiveness of $\mathcal G$, we derive that
    \[
        -\frac{\delta y}{2} \|Q_k^\ast - Q_k\|_\infty \le s (Q_k^\ast(u) - Q_{k + 1}(u)) \le \|Q_k^\ast - Q_k\|_\infty.
    \]
    Observe that, as $y \le \frac{2}{\delta}$, we can translate $s(Q_k^\ast - Q_k)$ by
    \[
        \tau := -s \frac{2 - \delta y}{4} \|Q_k^\ast - Q_k\|_\infty
    \]
    in order to improve the bounds, and find
    \[
        |Q_k^\ast(u) + \tau - Q_{k + 1}(u)| \le \frac{2 + \delta y}{4} \|Q_k^\ast - Q_k\|_\infty.
    \]
    Hence, as $\mathcal L^\ast$ is closed under translation, translating the fixed-point $Q_k^\ast$ by $\tau$
    yields
    \begin{align*}
        \|Q_{k+1}^\ast - Q_{k+1}\|_\infty &\le
        \|Q_k^\ast + \tau - Q_{k + 1}\|_\infty
        \le \frac{2 + \delta y}{4} \|Q_k^\ast - Q_k\|_\infty.
    \end{align*}

    \emph{Case 2}: Assume that $|\int_0^1 Q_k^\ast(v) - Q_k(v)\, dv| \le (1 - y) \|Q_k^\ast - Q_k\|_\infty$.
    Using first \eqref{eq:taylor}, then \eqref{eq:densitybounds}, the triangle inequality and that $\epsilon \in (0,1)$, and finally the assumption, we deduce
    \begin{align*}
        \|Q_k^\ast - Q_{k + 1}\|_\infty &= \lambda\text{-}\esssup_{u \in (0,1)} \left| \int_0^1 D_{\tilde Q_k}\mathcal G_u(v)(Q_k^\ast(v) - Q_k(v)) \, dv \right|
        \\
        &\le \lambda\text{-}\esssup_{u \in (0,1)} \left| \int_0^1 D_{\tilde Q_k} \mathcal G_u(v) - \epsilon \, dv \right| \|Q_k^\ast - Q_k\|_\infty + \epsilon \left| \int_0^1 Q_k^\ast(v) - Q_k(v) \, dv \right|
        \\
        &\le (1 - \epsilon y) \|Q_k^\ast - Q_k\|_\infty.
    \end{align*}
    By combining \emph{Case 1} and \emph{Case 2}, we have
    \[
        \|Q_{k+1}^\ast - Q_{k + 1} \|_\infty \le \max \Big( \frac{2 + \delta y}{4}, 1 - \epsilon y \Big) \| Q_k^\ast - Q_k \|_\infty,
    \]
    which is minimal for $y^\ast = \min( 2 / (\delta + 4\epsilon), 1)$ under the constraint that $y^\ast \in [0,\min(2/\delta,1)]$.
    Consequently, we conclude that $q$ given by
    \[
        q := 
        \begin{cases}
            \frac{\delta + 2\epsilon}{\delta + 4\epsilon} & 2 < \delta + 4\epsilon, \\
            1 - \epsilon & \text{else},
        \end{cases}
    \]
    is a suitable factor for \eqref{eq:thm.conv1.toshow}.

    \ref{it:conv_2}: For every $k \in \NN$, we have $Q_k \in \{ Q \in \QF : \|Q\|_\infty \le R\}$.
    Consequently, as in the final part of the proof of Lemma \ref{auxiliary_lemma}, we obtain that $(Q_k)_{k \in \NN}$ is relatively compact in $L^0(0,1)$.
    Let $(Q_{k_n})_{n \in \NN}$ be a convergent subsequence in $L^0(0,1)$ with limit $\hat Q$.
    By lower semicontinuity of $Q \mapsto \|Q\|_\infty$ on $L^0(0,1)$ and non-expansiveness of $\Gcal$, we get
    \[
        \|\hat Q - Q^\ast_k\|_\infty \le \liminf_{n \to \infty} \|Q_{k_n} - Q^\ast_k\|_\infty \le \|Q^\ast_k - Q_k\|_\infty \le q^k \| Q_0^\ast - Q\|_\infty \le q^k \|\hat Q - Q\|_\infty.
    \]
    In particular, we have
    \[
        \| \hat Q - Q_k\|_\infty \le \| \hat Q - Q_k^\ast\|_\infty + \|Q_k^\ast - Q_k\|_\infty \le 2 q^k \|\hat Q - Q\|_\infty,
    \]
    which is what we wanted to show.
\end{proof}

\begin{proposition}
    \label{prop:lin}
    Under Assumption \ref{ass:A1&2}, let $(\mu,\nu)$ be irreducible and $F \in \CDF$ such that one of the following conditions holds:
    \begin{enumerate}[label = (\roman*)]
        \item \label{it:prop.lin.1} $\mu,\nu$ and the measure with CDF $F$ are all symmetric;
        \item \label{it:prop.lin.2} $\mu$ is concentrated on finitely many points.
    \end{enumerate}
    Then $(\Acal^k F)_{k \in \NN}$ converges linearly.
\end{proposition}

\begin{proof}
    \ref{it:prop.lin.1}
    Let us assume that $\mu,\nu$ as well as the distribution of $F$ are symmetric.
    We have established in Lemma \ref{lem:Acal.properties} that then the distribution of $\Acal F$ is symmetric with the same median as $F$.
    Since $\Acal F \in \CDFc$, its mean exists and coincides by symmetry with its median.
    To summarize, we found that
    \[
        \int_0^1 Q_k(v) dv = \int_0^1 Q(v) \, dv \quad \text{for all }k \in \NN.
    \]
    Let $Q^\ast \in \Lcal^*$ be such that $\int_0^1 Q^\ast(v)\, dv = \int_0^1 Q(v) \, dv$.
    As in the proof of Theorem \ref{thm:conv}, we find, for every $k \in \NN$, some $t \in (0,1)$ and $\tilde Q_k = (1-t)Q_k + tQ^\ast$ such that \eqref{eq:taylor} and \eqref{eq:densitybounds} hold.
    As in case 2 of the proof of Theorem \ref{thm:conv}, we deduce that
    \[
        \|Q_{k + 1} - Q^\ast\|_\infty \le (1 - \epsilon) \| Q_k^\ast - Q_k \|_\infty + \epsilon \left| \int_0^1 Q_k(v) - Q^\ast(v) \, dv \right| = (1 - \epsilon) \| Q_k^\ast - Q_k \|_\infty.
    \]
    Hence, the rate of convergence of $(Q^k)_{k \in \NN}$ to $Q^\ast$ is at least linear.

    \ref{it:prop.lin.2}: Next, we assume that $\mu$ is concentrated on finitely many points, i.e.,  there are $n \in \NN$, $x_1,\ldots,x_n \in \RR$, and $w_1,\ldots,w_n \in (0,1]$ with $x_1 \le \dots \le x_n$ such that
    $\mu = \sum_{j = 1}^n w_j \delta_{x_j}$.
    Write $p_1 = 0$, $p_j = \sum_{m = 1}^{j-1} w_j$ for $j = 2,\ldots,n$, then we have for $u \in (0,1)$
    \[
        Q_\mu(u) = \sum_{j = 1}^n x_j \mathbbm 1_{[p_j,p_{j + 1})}(u).
    \]
    Because $\Gcal Q = S_Q^{-1} \circ Q_\mu$ (cf.\ Proposition \ref{mainProposition}), we also have that
    \[
        Q_k(u) = \Gcal^k Q(u) = \sum_{j = 1}^n y_j^{(k)} \mathbbm 1_{[p_j,p_{j + 1})}(u),
    \]
    for some $y_1^{(k)} ,\ldots,y_n^{(k)} \in \RR$ with $y_1^{(k)} \le \ldots \le y_n^{(k)}$.
    Similarly, there are $y_1^\ast,\ldots,y_n^\ast \in \RR$ with $y_1^\ast \le \ldots \le y_n^\ast$ such that $Q^\ast = \sum_{j = 1}^n y_j^\ast \mathbbm 1_{[p_j,p_{j + 1})}$ where we write $Q^\ast$ for the limit of $(Q_k)_{k \in \NN}$ in $L^\infty(0,1)$.
    Write $y^{(k)}$ resp.\ $y^\ast$ for the vector $(y^{(k)}_1,\ldots,y^{(k)}_n)$ resp.\ $(y^\ast_1,\ldots,y^\ast_n)$.
    Fix $k \in \NN$ and consider the translation
    \[
        \tau := \frac{\min(y^{(k)} - y^\ast) + \max(y^{(k)} - y^\ast)}2.
    \]

    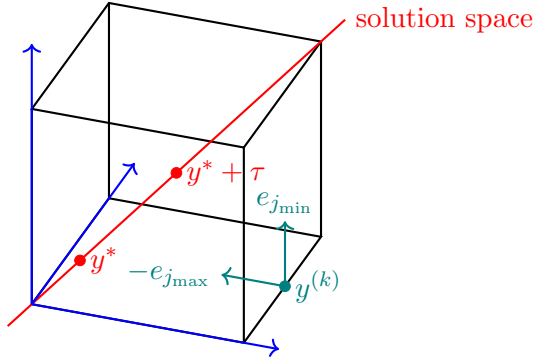
\begin{figure}[H]
     \centering
    \tdplotsetmaincoords{60}{20}
     \begin{tikzpicture}
        [tdplot_main_coords, cube/.style={thick,black},
        solutionSpace/.style={thick,red},
        axis/.style={->,blue,thick},
        vector/.style={->,teal,thick},]
    
        \draw[cube] (0,0,0) -- (0,3,0) -- (3,3,0) -- (3,0,0) -- cycle;
        \draw[cube] (0,0,3) -- (0,3,3) -- (3,3,3) -- (3,0,3) -- cycle;
    
        \draw[cube] (0,0,0) -- (0,0,3);
        \draw[cube] (0,3,0) -- (0,3,3);
        \draw[cube] (3,0,0) -- (3,0,3);
        \draw[cube] (3,3,0) -- (3,3,3);
    
        \draw[solutionSpace] (-0.25,-0.25,-0.25) -- (3.25,3.25,3.25)node[anchor=west]{solution space};
        \draw[axis] (0,0,0) -- (0,0,4) node[anchor=west] {};
        \draw[axis] (0,0,0) -- (0,4,0) node[anchor=east] {};
        \draw[axis] (0,0,0) -- (3.5,0,0) node[anchor=west] {};
    
        \filldraw [red] (1.5,1.5,1.5) circle(2pt) node[anchor=west]{$y^*+\tau$};
        \filldraw [red] (0.5,0.5,0.5) circle(2pt) node[anchor=west]{$y^*$};
    
        \filldraw [teal] (3,1.6,0) circle(2pt) node[anchor=west]{$y^{(k)}$};
        \filldraw [teal] (3,1.6,1.6) circle(0pt) node[anchor=north]{$e_{j_{\min}}$};
        \draw[vector] (3,1.6,0) -- (3,1.6,1) node[anchor=west]{};
        \draw[vector] (3,1.6,0) -- (2.1,1.6,0) node[anchor=east]{$-e_{j_{\max}}$};
    \end{tikzpicture}
    \caption{Geometric visualization of the translation in the proof of Proposition~\ref{prop:lin}\ref{it:prop.lin.2}, where we write $e_i$ for unit vectors in $\RR^n$ as well as $j_{\max} = \arg \max_j(y^{(k)}_j - y^\ast_j)$.}
    \end{figure}
    On the one hand, we have
    \begin{equation}
        \label{eq:tau1}
        \|y^{(k)} - y^\ast - \tau\|_\infty = \frac12 \left( \max(y^{(k)} - y^\ast) - \min(y^{(k)} - y^\ast) \right).        
    \end{equation}
    On the other hand, we can use $\|Q_k - Q^\ast\|_\infty = \|y^{(k)} - y^\ast\|_\infty$ to derive from non-expansiveness of $\Gcal$ that 
    \begin{equation}
        \label{eq:tau2}
        |\tau| = \|y^\ast - y^\ast -\tau\|_\infty = \lim_{\ell \to \infty} \|y^{(\ell)} - y^\ast - \tau \|_\infty \le \|y^{(k)} - y^\ast - \tau \|_\infty.
    \end{equation}
    By plugging \eqref{eq:tau1} into \eqref{eq:tau2}, we find the following inequalities:
    \begin{align*}
        \min_j(y^{(k)}_j - y^\ast_j) \le 0 \text{ and } \max_j(y^{(k)}_j - y^\ast_j) \ge 0.
    \end{align*}
    By Taylor's theorem \cite[Section 5.6]{Ca64} we find $t \in [0,1]$ and $\tilde Q_k := (1-t)Q_k + t Q_\ast$ such that \eqref{eq:taylor} holds.
    Therefore, we have for $\lambda$-almost every $u$ that
    \begin{align*}
        D_{\tilde Q_k} \Gcal_u (Q_k - Q^\ast) 
        &= 
        \sum_{j = 1}^n \int_{p_j}^{p_{j + 1}} D_{\tilde Q_k} \Gcal_u(v) \, dv \, (y^{(k)}_j - y^\ast_j) \\
        &\le \sum_{j = 1, j \neq j_{\min}}^n \int_{p_j}^{p_{j + 1}} D_{\tilde Q_k} \Gcal_u(v) \, dv \| y^{(k)} - y^\ast \|_\infty \\
        &\le (1 - \epsilon \min_j(w_j)) \| y^{(k)} - y^\ast \|_\infty,
    \end{align*}
    where $w = (w_1,\ldots,w_n)$, $j_{\min} = \arg \min_j(y^{(k)}_j - y^\ast_j)$ and $\epsilon$ is some lower bound for the density as in \eqref{eq:densitybounds}.
    For symmetry reasons, we obtain
    \[
        \|y^{(k+1)} - y^\ast\|_\infty = \| Q_{k + 1} - Q^\ast \|_\infty \le \|D_{\tilde Q_k}(Q_k - Q^\ast)\|_\infty \le (1 - \epsilon \min_j(w_j)) \| y^{(k)} - y^\ast \|_\infty,
    \]
    which yields that the convergence rate of $y^k \to y^\ast$ w.r.t.\ $\|\cdot\|_\infty$ is at least linear.
\end{proof}

\begin{remark}
    The precise convergence rates obtained in the proof of Theorem \ref{thm:conv} depend on upper and lower bounds on the density of the considered Fréchet derivatives, in the sense of Remark \ref{rmk:elliptic}.
    We have shown in Proposition \ref{mainProposition} that $Q \mapsto \left( (u,v) \mapsto D_Q \Gcal_u(v) \right)$ is continuous as a function from $L^\infty(0,1)$ to $L^\infty((0,1)^2)$.
    Therefore, asymptotically our obtained rates can be improved by using upper and lower bound of the density of $D_{Q^\ast} \Gcal$, that are
    \[
        \lambda^2\text{-}\essinf_{u,v \in (0,1)} D_{Q^\ast} \Gcal_u(v) \text{ and }
        \lambda^2\text{-}\esssup_{u,v \in (0,1)} D_{Q^\ast} \Gcal_u(v).
    \]
\end{remark}

\section{Computational aspects and numerical results}\label{sect:numerics}
In our numerical experiments we consider irreducible pairs $(\mu, \nu)$ satisfying Assumptions \ref{ass:A1&2}. Specifically, we discretize $\mu$ via quantization, obtaining a distribution $\mu_n=\frac{1}{n}\sum_{i=1}^n \delta_{x_i}$ such that $\mu_n \preceq_c \mu$ for $n \in \NN$, where $x_i = \int_{i/n}^{(i+1)/n} Q_\mu(u)du, i = 0, \dots, n-1$. Subsequently, we fix $n$ for each experiment and compute the standard stretched Brownian motion from $\mu_n$ to $\nu$. Notably, although the pair $(\mu_n, \nu)$ remains irreducible, the pair $(\mu_n, \mu)$ is non-irreducible. This is because $\RR$ can be partitioned into $n$ irreducible components due to the construction of $\mu_n$.

Finally, it is worth mentioning that we opt for the semi-discrete setting. This choice is motivated by the stability of the Martingale Benamou-Brenier Problem (see \cite{BeJoMaPa21b}), and the ease with which the $\infty$-Wasserstein distance between two discrete measures can be computed.

\paragraph{{\color{black} Linear convergence.}}
In our experiments we observe a rapid convergence of the fixed-point scheme, often reaching the fixed-point solution within just a few iterations. The required number of iterations for achieving convergence appears to be influenced by the proximity of our problem pair to a non-irreducible pair, as discussed below.

In particular, Figure~\ref{fig.lin} illustrates the linear convergence of the iteration scheme for the \ssbm  be\-tween $\mu_{50}$ and $\nu$, where $\mu$ is a convex combination of a normal and a logistic distribution, and $\nu$ is a truncated normal. The results indicate convergence in approximately $15$ iterations, regardless of the choice of the initial distribution for the fixed-point iteration. It is worth noting that the computation of the convolution with the Gaussian density $\phi$ plays a pivotal role in achieving this convergence. Therefore, ensuring a well-optimized implementation of this operation is crucial for obtaining optimal results.
\begin{figure}
\centering
\includegraphics[width=0.81\textwidth]{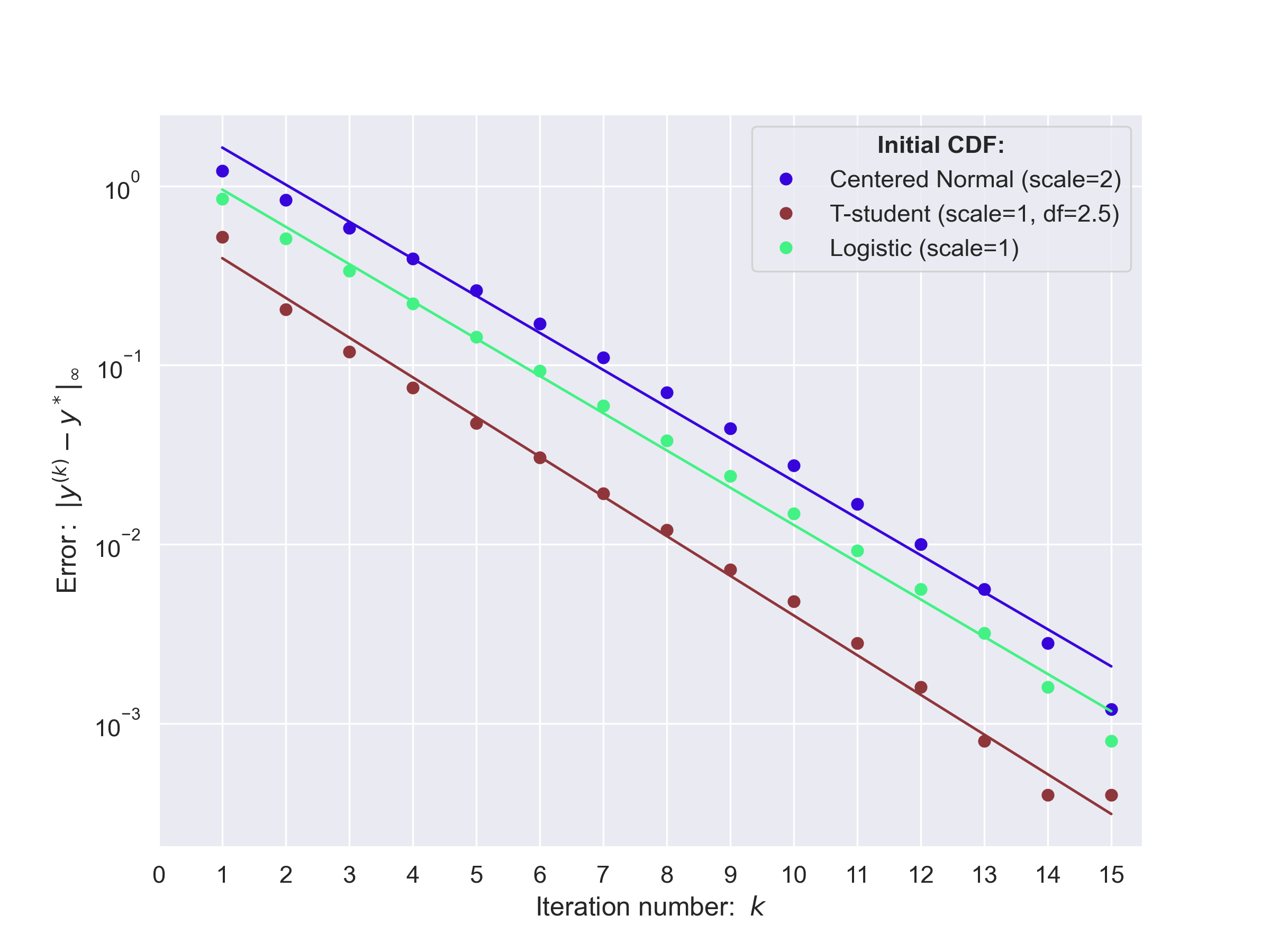}
\caption{Linear convergence in the fixed-point iteration across various initial CDFs.}\label{fig.lin}
\end{figure}

\paragraph{Irreducibility of the pair $(\mu, \nu)$.}
The irreducibility of the pair plays a critical role in determining the convergence of the algorithm. Specifically, when the pair $(\mu_n, \nu)$ is close to a non-irreducible pair, a higher number of iterations is required to reach the fixed-point solution. In Figure~\ref{fig.iter}, we observe the number of iterations needed to compute the \ssbm from $\mu_{10}$ to $\nu$, where both $\mu$ and $\nu$ are truncated normal distributions. Notably, while the standard deviation of $\mu$ is fixed at $1$, the standard deviation of $\nu$ approaches $1$.
Figure~\ref{fig.leb} illustrates the dependency of the size of the support of the fixed-point solution on the standard deviation of $\nu$.
It is noteworthy that consistent results are obtained when the standard deviation of $\nu$ is fixed, and the standard deviation of $\mu$ approaches that of $\nu$.

\begin{figure}
    \centering
    \begin{subfigure}[b]{0.78\textwidth}
        \includegraphics[width=1\linewidth]{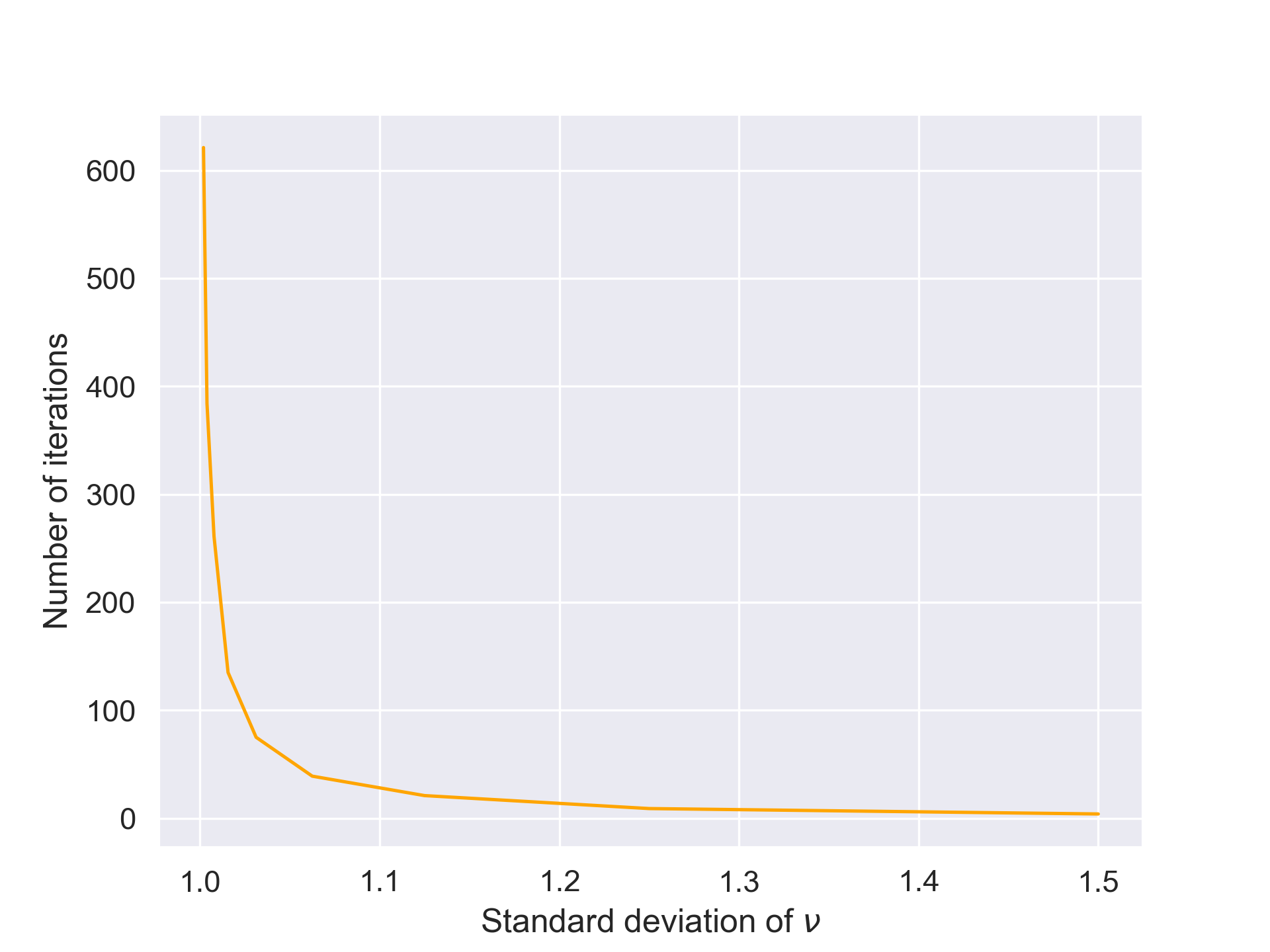}
        \caption{Number of iterations before convergence} \label{fig.iter}
    \end{subfigure}
    \begin{subfigure}[b]{0.78\textwidth}
    \includegraphics[width=1\linewidth]{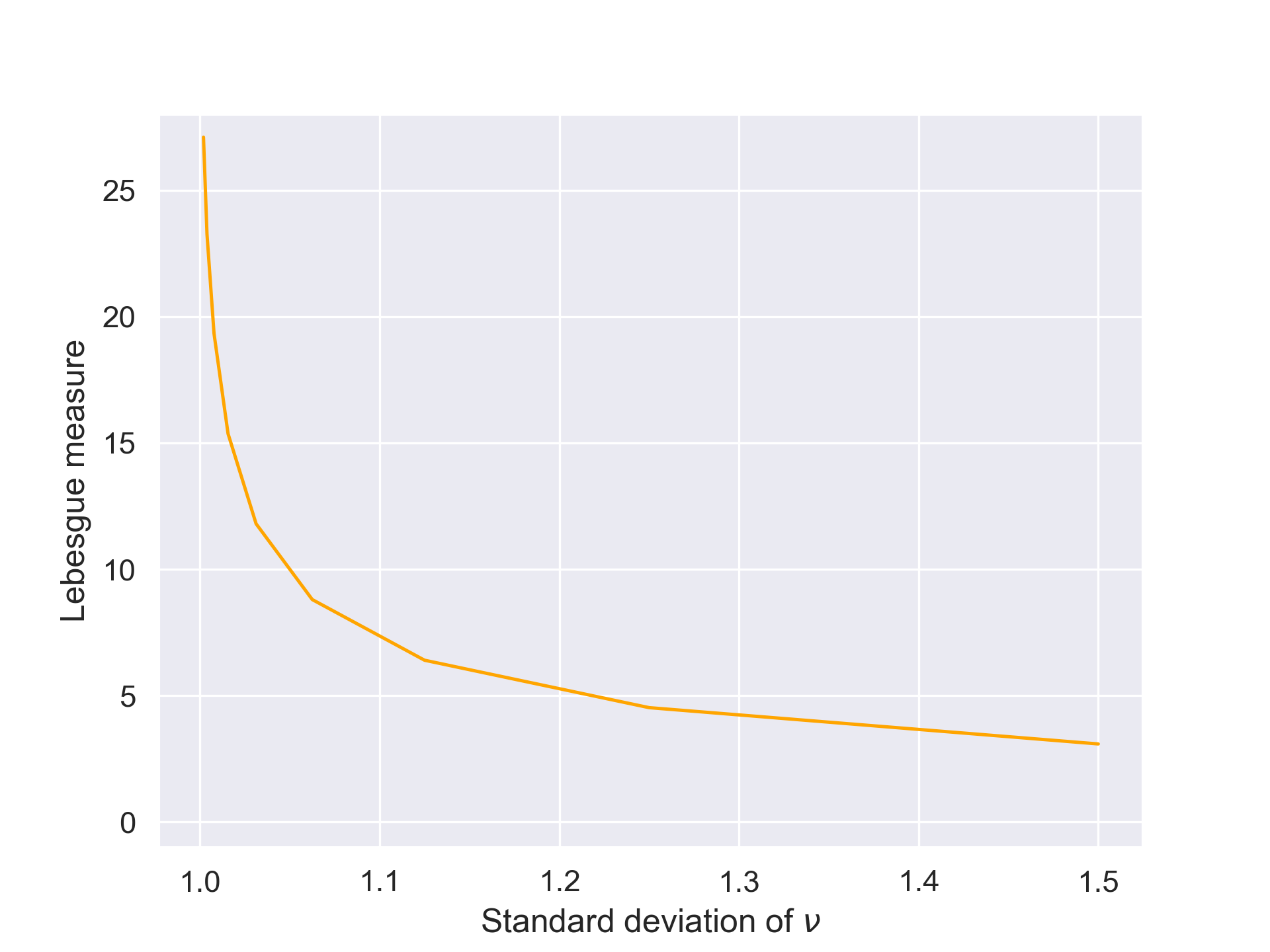}
        \caption{Lebesgue measure of the convex hull of the support of the fixed-point solution} \label{fig.leb}
    \end{subfigure}
    \caption{Computation of the standard stretched Brownian motion from $\mu_{10}$ to $\nu$, where both $\mu$ and $\nu$ are truncated normal distributions with $\mu$ having a fixed standard deviation of $1$.}
    \end{figure}

\paragraph{Semidiscrete system.}
We conclude by noticing that when $\mu$ is concentrated on a finite number of points, the fixed-point equation $\mathcal{A}F=F$ is equivalent to a system of nonlinear equations. This property suggests, therefore, an alternative method for calibrating the model.

\begin{corollary}[Semidiscrete system]
\label{remark_system}
    Let $\mu \in \Pcal(\RR)$ be concentrated on finitely many points, i.e. $\mu=\sum_{j=1}^n w_j \delta_{x_j}$, where $x_1, \dots, x_n \in \RR$, $w_1, \dots, w_n > 0$ and $w_1 + \dots + w_n = 1$. Then solving the fixed-point equation $\Acal F = F$ is equivalent to solving the following system of nonlinear equations
    \begin{equation}
		\label{nonLinearSyst}
			\int_{\mathbb R} Q_\nu \left (\sum_{j=1}^n w_j\Phi\left(y_i-y_j-z\right) \right) \phi(z)dz=x_i\;\;\; \text{for } i \in \{1, \dots n\}.
	\end{equation}
    In particular, $F \in \CDF$ solves the fixed-point equation if and only if it is the CDF of $\sum_{j=1}^n w_j \delta_{y_j}$ and $(y_1, \dots, y_n)$ is the solution to \eqref{nonLinearSyst}.
\end{corollary}
\begin{proof}
    The proof follows from Proposition \ref{mainProposition} \ref{it:mainProposition.2}.
\end{proof}

\begin{example}
	Let $\mu = \sum_{j=1}^n w_j \delta_{x_j}$ be a symmetric probability distribution, $\nu = \text{Unif}_{[0,1]}$, and $\mu \preceq_c \nu$. Then, as a consequence of Corollary \ref{remark_system}, the unique (up to a shift constant) fixed point of $\Acal$ is the CDF of $\sum_{j=1}^n w_j \delta_{y_j^*}$ such that
	\begin{equation*}
		\sum_{j=1}^n w_j\Phi_2\left(\frac{y_i^*-y_j^*}{\sqrt{2}}\right) =x_i\;\;\; \text{for any } i \in \{1, \dots n\},
	\end{equation*}
 where $\Phi_2$ is the CDF of $\gamma_2$. For instance, let us take $n = 4$. By removing the 4th equation and using the symmetry of $\Phi_2$, we have
	 \begin{equation*}
		\begin{cases}
		w_2\Phi_2\left(\frac{y_1^*-y_2^*}{\sqrt{2}} \right)+w_3\Phi_2\left(\frac{y_1^*-y_3^*}{\sqrt{2}}\right)+w_4\Phi_2\left(\frac{y_1^*-y_4^*}{\sqrt{2}} \right) = 4x_1-w_1/2 \\
		-w_1\Phi_2\left(\frac{y_1^*-y_2^*}{\sqrt{2}} \right)+w_3\Phi_2\left(\frac{y_2^*-y_3^*}{\sqrt{2}}\right) +w_4\Phi_2\left(\frac{y_2^*-y_4^*}{\sqrt{2}} \right) = 4x_2-w_2/2-w_1 \\
		-w_1\Phi_2\left(\frac{y_1^*-y_3^*}{\sqrt{2}} \right)-w_2\Phi_2\left(\frac{y_2^*-y_3^*}{\sqrt{2}}\right) +w_4\Phi_2\left(\frac{y_3^*-y_4^*}{\sqrt{2}} \right) = 4x_3-w_3/2-w_1-w_2\\
		\end{cases}.
	\end{equation*}
	Moreover, we can suppose that $y_4^*=1-y_1^*$, $y_3^*=1-y_2^*$ by using Lemma \ref{lem:Acal.properties} \ref{it:A.symmetry}, and we can fix $y_1=0$ by shift-invariance (Lemma \ref{lem:Acal.properties} \ref{it:A.shift}). Then we get
	\begin{equation*}
		w_2\Phi_2\left(-\frac{y_2^*}{\sqrt{2}}\right)+w_3\Phi_2\left(\frac{y_2^*-1}{\sqrt{2}}\right)= 4x_1-w_1/2-\Phi_2\left(-\frac{1}{\sqrt{2}}\right).
	\end{equation*}
\end{example}
\appendix

\section{Approximation results}

\begin{lemma}
\label{lem:approx.corol}
Let $\mu, \nu \in \mathcal P_1(\mathbb R)$ with $\supp(\mu) \subseteq {\rm int}({\rm co}({\rm supp}(\nu)))$. Then there exists a sequence of pairs $(\mu_n, \nu_n)$, $n \in \NN$ that satisfy Assumption \ref{ass:A1&2} and 
\begin{equation*}
\lim_{n \to \infty} \mathcal W_1(\mu_n,\mu) = 0 = \lim_{n \to \infty} \mathcal W_1(\nu_n,\nu).
\end{equation*}
\end{lemma}
\begin{proof}
    By standard approximation results (for example via quantization) it is well-known that any measure with finite first moment admits an $\mathcal W_1$-approximating sequence by measures concentrated on finitely many points.
Let $(\mu_n)_{n \in \NN}$ (resp. $(\hat \nu_n)_{n \in \NN}$) be such a sequence for $\mu$ (resp. $\nu$). Furthermore, we can assume without loss of generality that $\supp(\mu_n) \subseteq {\rm int}({\rm co}({\rm supp}(\hat \nu_n)))$, for any $n \in \NN$.
Let $R_n := \max\{ |x| : x \in \supp(\hat \nu_n) \}$ and, for $R > 0$, write $\lambda_{R}$ for the uniform distribution on $[-R,R]$.
Then the measure 
\begin{equation}
\label{approxNuDef}
    \nu_n := \frac{1}{n R_n} \lambda_{R_n} + \left( 1 - \frac{1}{n R_n} \right) \left(\hat \nu_n \ast \lambda_{1/n}\right)
\end{equation}
is absolutely continuous w.r.t.\ the Lebesgue measure, with density bounded away from zero on $\co(\supp(\nu_n)) = \supp(\nu_n)$.
Thus, it is straightforward to verify that $\nu_n \to \nu$ in $\Wcal_1$, $\mu_n := \hat \mu_n \to \nu$ in $\Wcal_1$.
\end{proof}

\begin{lemma}[Existence of approximating sequence]
\label{lem:approx.sequence}
Let $\mu, \nu \in \mathcal P_1(\mathbb R)$ with $\mu \preceq_c \nu$. 
Then there exists a sequence of pairs $(\mu_n,\nu_n)$, $n \in \NN$, that are irreducible, satisfy the Assumption \ref{ass:A1&2}, and
\begin{equation}
\label{eq:lem.approx.sequence.assertion}
\lim_{n \to \infty} \mathcal W_1(\mu_n,\mu) = 0 = \lim_{n \to \infty} \mathcal W_1(\nu_n,\nu).
\end{equation}
\end{lemma}

\begin{proof}
Let us consider the sequence $(\nu_n)_{n \in \NN}$ defined in \eqref{approxNuDef}. We define $\hat \mu_n$ as the so-called Wasserstein projection in the convex order of $\mu$ onto $\{ \eta \in \mathcal P_1(\mathbb R) : \eta \preceq_c \hat \nu_n \}$, see \cite{AlCoJo20, BaBePa19}.
In \cite[Theorem 1.1]{JoMaPa23} the following estimate was established
\[
\Wcal_1(\mu, \hat \mu_n) \le \Wcal_1(\nu, \hat \nu_n),
\]
whence $\hat \mu_n \to \mu$ in $\Wcal_1$.
Finally, we set
\[
    \mu_n := \frac{1}{n R_n} \delta_0 + \left(1 - \frac{1}{n R_n} \right) \hat \mu_n.
\]
To conclude the assertion, it remains to show that $(\mu_n,\nu_n)$ is irreducible.
A simple application of Jensen's inequality reveals that we have, for every $x \in \interior(\supp(\nu))$,
\[
    \int_\RR |x - y| \, \mu_n(dy) < \int_\RR |x - y| \, \nu_n(dy).
\]
Hence, irreducibility follows from \cite[Definition A.3]{BeJu16}.
\end{proof}

\begin{lemma} \label{lem:aux.ToP}
Let $T \colon (0,1) \to \RR$ be measurable and bounded, and $(\mu_k)_{k \in \NN}$ be a sequence in $\mathcal P(\RR)$ that converges in total variation to $\mu$, with $\mu_k \ll \mu$.
Then
\begin{equation*}
T \circ F_{\mu_k} \to T \circ F_\mu \quad \text{in } L^1(\RR; \mu)
\end{equation*}
\end{lemma}

\begin{proof}
Assume without loss of generality that $T$ is bounded by 1.
By Lusin's theorem, for every $\epsilon > 0$ there is a compact set $K \subseteq (0,1)$ with $\lambda(K) \ge 1 - \epsilon$ and $T|_K$ is continuous.
We use Tietze's extension theorem to find a continuous function $T^K \colon \RR \to [-1,1]$ with $T^K|_K = T|_K$.
We compute
\begin{align*}
\int_\RR |T &\circ F_{\mu_k}(x) - T \circ F_\mu(x)| \, d\mu(x) \\
&\le \int_{\RR} |T^K \circ F_{\mu_k}(x) - T^K \circ F_\mu(x)| \, d\mu(x)
+ \int_{Q_{\mu_k}(K^c)} |T^K \circ F_{\mu_k}(x) - T \circ F_{\mu_k}(x)| \, d\mu(x) \\
&\phantom{\le \int_{\RR} |T^K \circ F_{\mu_k}(x) - T^K \circ F_\mu(x)| \, d\mu(x)}\; + \int_{Q_\mu(K^c)} |T^K \circ F_{\mu}(x) - T \circ F_{\mu}(x)| \, d\mu(x).
\end{align*}
Clearly, the first term of the right-hand side vanishes for $k \to \infty$.
We proceed to estimate the last two terms.
Using that $T$ and $T^K$ are bounded by 1, we obtain
\begin{align*}
\limsup_{k \to \infty} \int_\RR | T\circ F_{\mu_k}(x) - T\circ F_\mu(x)| \, d\mu(x) &\le
\limsup_{k \to \infty} 2 \mu(Q_{\mu_k}(K^c)) + {2}\mu(Q_\mu(K^c)) \\
&\le \limsup_{k \to \infty}2 \mu(Q_{\mu_k}(K^c)) + {2}\epsilon.
\end{align*}
Denoting by $f_k$ the density of $\mu_k$ w.r.t. $\mu$, we have
\begin{align*}
\mu(Q_{\mu_k}(K^c)) &\le \int_{Q_{\mu_k}(K^c)} f_k(x) \, d\mu(x) + \int_\RR |1 - f_k(x)| \, d\mu(x) \\
&\le \mu_k(Q_{\mu_k}(K^c)) + TV(\mu,\mu_k) \le \epsilon + TV(\mu,\mu_k).
\end{align*}
Therefore, we find
\[
\limsup_{k \to \infty} \int_\RR | T\circ F_{\mu_k}(x) - T\circ F_\mu(x)| \, d\mu(x) \le 4 \epsilon,
\]
and conclude as $\epsilon > 0$ was arbitrary.
\end{proof}

\bibliographystyle{abbrv}
\bibliography{joint_biblio}
\end{document}